\documentclass[journal]{IEEEtran}

\usepackage{xcolor,soul,framed} 
\usepackage{bm}
\usepackage[font=small,labelfont=bf]{caption}
\colorlet{shadecolor}{yellow}
\usepackage[pdftex]{graphicx}
\graphicspath{{../pdf/}{../jpeg/}}
\DeclareGraphicsExtensions{.pdf,.jpeg,.png}
\usepackage{import}
\usepackage{amsmath}
\usepackage{amssymb}
\usepackage[linktocpage=true]{hyperref}
\usepackage{hyperref}
\usepackage{multicol,lipsum}
\usepackage{array}
\usepackage{xcolor}
\usepackage{makecell}
\usepackage{tikz}
\usepackage{blkarray}
\usepackage{amsmath,amsfonts,amsthm,amssymb}
\usepackage{mathtools}
\usepackage{cuted}
\usepackage{cite}
\usepackage{cleveref}
\usepackage{amssymb}
\usepackage{cases}
\usepackage{empheq}
\usepackage{mdframed}

\newcommand{\tr}{\textcolor{red}}

\newcommand{\mb}{\mathbf}
\newcommand{\mbb}{\mathbb}
\newcommand{\N}{\mathsf{N}}
\newcommand{\x}{\bm{x}}
\newcommand{\1}{\mathtt{A}}
\newcommand{\2}{\mathtt{B}}

\newcommand{\p}{\mathsf{p}}
\renewcommand{\pi}{\eta}
\newcommand{\y}{\bm{y}}

\newcommand{\V}{\mathcal{V}}
\newcommand{\mi}{\mathit}
\newcommand{\X}{\mathcal{X}}
\def\mathbi#1{\textbf{\em #1}}
\newtheorem{theorem}{Theorem}
\newtheorem{lemma}{Lemma}
\newtheorem{assumption}{Assumption}
\newtheorem{remark}{Remark}
\newtheorem{cor}{Corollary}
\newtheorem{definition}{Definition}
\newtheorem{prop}{Proposition}
\newtheorem{example}{Example}

\newtheorem{revisit}{Revisiting Example}

\newtheorem{lem}{Lemma}

\newtheorem{appendixCorollary}{Corollary}

\newtheorem{thm}{Theorem}

\begin{document}
\bstctlcite{IEEEexample:BSTcontrol}
       \title{From Discrete to Continuous {Binary} Best-Response Dynamics: Discrete Fluctuations {Almost Surely Vanish with} Population Size}
   \author{
       Azadeh Aghaeeyan 
       and~Pouria~Ramazi

  \thanks{
  This  paper {was presented in part at} the American Control Conference, Toronto, Canada,  July  8-12,  2024. 
  This work was completed while A. Aghaeeyan and P. Ramazi were with the Department of Mathematics and Statistics, Brock University, St. Catharines, ON L2S 3A1, Canada (e-mail: $\{$a.aghaeeyan, p.ramazi$\}$@gmail.com).
   This work has been submitted to the IEEE for possible publication. Copyright may be
transferred without notice.
}%
 }


\maketitle

\begin{abstract}
In binary decision-making, individuals often go for a common or a rare action.
In the framework of evolutionary game theory, the best-response update rule can be used to model this dichotomy. 
Those who prefer the common action are called \emph{coordinators} or \emph{conformists}, and those who prefer the rare one are called \emph{anticoordinators} or \emph{nonconformists}.
A finite mixed population of the two types may undergo perpetual fluctuations, the characterization of which appears to be challenging.
It is particularly unknown whether the fluctuations persist as the population size grows. 
To fill this gap, we approximate the discrete population dynamics of coordinators and anticoordinators with the associated mean dynamics in the form of differential inclusions. 
We show that the family of state sequences of the discrete dynamics with increasing population sizes forms a generalized stochastic approximation process for the differential inclusion. 
On the other hand, we show that the differential inclusion always converges to an equilibrium. 
This implies that the reported perpetual fluctuations in the discrete dynamics of coordinators and anticoordinators 
almost surely vanish with population size. 
The results motivate analyzing the often simpler mean dynamics, which partly reveal the asymptotic behavior of the discrete dynamics. 

\end{abstract}


%
\IEEEpeerreviewmaketitle


\section{Introduction}
Humans face a variety of repetitive decision-making problems, such as whether to follow the fashion trends, get a flu shot, or sign a petition \cite{cao2013fashion,blanton2001introduction}.
In such two-option decision-making problems, individuals are often either \emph{coordinators}, also known as \emph{conformists}, those who adopt a decision whenever it has already been adopted by a particular population proportion of decision-makers,
or \emph{anticoordinators}, also known as \emph{nonconformists}, those who go for a decision whenever it has been rejected by a particular population proportion \cite{coordinationandanticoordination, arditt2}.

In the framework of evolutionary game theory, the so-called  \emph{best-response} update rule can capture the decision-making processes of coordinators and anticoordinators by assuming affine payoff functions.
 An increasing payoff function with respect to the commonality of the decisions makes the best-response update rule equivalent to the \emph{linear threshold model}, where individuals are coordinators with possibly unique thresholds \cite{granovetter1978threshold}.
 A decreasing payoff function, on the other hand, models the anticoordinators.

Analyzing the outcome of these decision-making processes requires accounting for each individual's actions, and consequently,
 a large body of research studied the steady state behavior of finite populations of best-responders.
The problem has been investigated in various settings, ranging from well-mixed \cite{zeroDeterminantS} to structured populations \cite{ZHU2023110707, arefizadeh2023robustness,arditt2},
 and from time-invariant to time-varying thresholds \cite{arditti2024robust}.
Semi-tensor products were also utilized to {analyze networked evolutionary games}
 \cite{guo2013algebraic, cheng2014finite} { and  logical dynamic games \cite{boolean1,10505780} }.

 A finite population of heterogeneous coordinators, who differ in perceiving the commonality, admits equilibria \cite{vanelli2020games} and equilibrates under asynchronous activation sequence  \cite{ramazi2020convergence}. 
A finite population of heterogeneous anticoordinators either equilibrates or converges to a two-state cycle \cite{ramazi2017asynchronous, grabisch2020anti}.
As for finite populations of coordinating and anticoordinating agents, it was shown that the dynamics may undergo perpetual fluctuations  \cite{grabisch2019model,grabisch2020anti, roohi}.
The exact form and the conditions for the existence of the fluctuations appear to be an open problem. 

Researchers get around the challenges arising from the analysis of finite population dynamics by exploiting the associated deterministic \emph{mean dynamics} \cite{sandholm2010population, nowak2019homogeneous}, where the population is assumed to be infinitely large.
In this regard, most studies investigated the replicator dynamics which are the mean dynamics of the other popular update rule, \emph{imitation} \cite{comoImitation, replicator, replicator2, ramazireplicator}.
Fewer studies considered the best-response mean dynamics, which are differential inclusions \cite{hofbauer1995stability, golman2010basins,bestresponsePotential,berger2007two,theodorakopoulos2012selfish}.
Some studies approximated the evolution of decisions in structured populations with ordinary differential equations \cite{comoApproximation,ravazzi2023asynchronous,peidynamic}.

Although analysis of the mean dynamics is generally more straightforward compared to that of finite populations, the potential discrepancy between these two may question the validity of the approximation.
Hence, much effort has been devoted to connecting the behavior of finite populations with the associated mean dynamics as the population size grows \cite{benaim2005stochastic,benaim1998recursive, benaim2003deterministic, benaim2009mean}. 
The results on infinite horizon behavior connect the {Birkhoff center} of the mean dynamics with the support of the stationary measures of the Markov chain corresponding to the finite population's dynamics \cite{benaim1998recursive, roth2013stochastic}.

How can these results, which link the mean dynamics to finite population dynamics, be applied to the heterogeneous populations of coordinators and anticoordinators?
To the best of our knowledge, no study has investigated the asymptotic behavior of infinite heterogeneous populations of all coordinators, all anticoordinators, or a mixture of both.
Even if there were such results, could they shed light on the discrete population dynamics?  

We provide the answer in this paper,
 namely by using the deterministic differential inclusions as an approximation tool to reveal the asymptotic behavior of the discrete population dynamics as the population size approaches infinity.
First, through an intuitive example, we show that the perpetual fluctuations do not scale as the population size grows.
Second, we write the population dynamics as a Markov chain and obtain the associated continuous-time population dynamics--an upper semicontinuous differential inclusion.
Third, we show that the collection of population dynamics Markov chains defines \emph{generalized stochastic approximation processes} for the continuous-time population dynamics \cite{roth2013stochastic}.
This allows us to apply the results in \cite{roth2013stochastic} to approximate the asymptotic behavior of the discrete population dynamics by finding the Birkhoff center of the continuous-time population dynamics, which we do in the last part of the paper.

Our contribution is fourfold:
\emph{i)}
We show that in the continuous-time population dynamics, the heterogeneous population of all coordinators may admit two types of equilibrium points, clean-cut and ruffled coordinator-driven, where clean-cut equilibria are asymptotically stable and ruffled are unstable--\textbf{\Cref{lem:eqMixedAbstract}} and \textbf{\Cref{thm:ROA}}.
On the contrary, the heterogeneous population of all anticoordinators admits only one globally asymptotically stable equilibrium point that could be either clean-cut or ruffled anti-coordinator driven--{\textbf{\Cref{lem:eqMixedAbstract}} and} \textbf{\Cref{lem: stabilityOfType01}}.
These results are in line with the existing results on the steady-state behavior of finite heterogeneous populations of all coordinators \cite{ramazi2020convergence} and all anticoordinators \cite{ramazi2017asynchronous}.
\emph{ii)}
We show that in the continuous-time population dynamics, the mixed heterogeneous population of coordinators and anticoordinators may admit all three types of equilibria, clean-cut, ruffled anticoordinator-driven, and ruffled coordinator-driven, where clean-cut and ruffled anticoordinator-driven equilibria are asymptotically stable, the other type is unstable, and the dynamics always converge to one of the equilibria--\textbf{\Cref{lem:eqMixedAbstract}} and \textbf{\Cref{thm:ROA}}.
\emph{iii)}
Consequently, we show that the reported perpetual fluctuations in the proportion of $\1$-players \cite{roohi} almost surely vanish with the population size --\textbf{\Cref{thm:2}} and \textbf{\Cref{cor_fluctuationsDoNotScaleWithN_discretePopulationDynamics_2}}.
\emph{iv)}
Out of the analysis of the continuous-time population dynamics, insightful hints on the asymptotic behavior of the finite populations were obtained, such as having the same equilibria and similar long-term behavior.
 Hence, in general, for other population dynamics, it may prove useful to study the behavior of the continuous-time population dynamics, which is usually more straightforward, prior to the precise analysis of the finite population dynamics.

\subsection*{Notations}
The sets of real, nonnegative real, integer, nonnegative integer, and positive integer numbers are respectively shown by $\mathbb{R}$, $\mathbb{R}_{\geq 0}$, $\mathbb{Z}$, $\mathbb{Z}_{\geq 0}$, and $\mathbb{N}$.
{With an abuse of notation, for positive integers  $n$ and $m$, define $\frac{1}{m}\mathbb{Z}^n = \{\frac{1}{m}\y\mid\y \in \mathbb{Z}^n\}$.}
For each $i \in \mathbb{N}$, $[i]$ denotes $\{1,2,\ldots,i\}$.
Boldface letters refer to vectors.
The notation $\bm 1$ refers to a vector of appropriate size with all elements equal to $1$.
The $i^{\text{th}}$ element of the vector $\mb{q}$ is indicated by ${q}_i$.
The calligraphic font $\mathcal{X}$ represents a set.
A sequence of variables $x_0, x_1, x_2,\ldots$ is represented by $\langle x_k\rangle{_{k=0}^\infty}$.
The Euclidean norm of a vector $\x$ is denoted by $\vert \x \vert$.
{By the notations $[a_1,b_1] - [a_2,b_2]$ and $[a_1,b_1] - {c}$ we mean $[a_1 - a_2, b_1 - b_2]$ and $[a_1-c,b_1-c]$, respectively. }
The notation $\text{cl}(\mathcal{X})$ denotes the closure of the set $\mathcal{X}$.
A set-valued map ${\bm {\mathcal{V}}}(\x)$ from $\mathbb{R}^n$ to $\mathbb{R}^n$ is denoted by notation ${\bm {\mathcal{V}}}: \mathbb{R}^n\rightrightarrows \mathbb{R}^n$.
\label{ch:intro}

\section{Problem Formulation} \label{sec:problemFormulation}
We consider a well-mixed population of $\N$ agents, labeled by $1,2,\ldots,\N$, choosing repeatedly between two strategies (or actions) $\1$ and $\2$ over a discrete time sequence $ t \in \frac{1}{\N} \mathbb{Z}_{\geq 0}$.
The time sequence is indexed by $k$ where $k =  \N t $ (i.e., the index of time $\frac{k}{\N}$ is $k$).
Each agent $i \in [\N]$ has a \emph{threshold} $\tau(i)\in (0,1)$ and is either a \emph{coordinator} (conformist) or an \emph{anticoordinator} (a nonconformist).
A coordinator $i$ (resp. an anticoordinator $i$) tends to choose strategy $\1$ whenever the population proportion of strategy-$\1$ players, excluding herself,  denoted $x^{\N}_{-i}\in[0,1)$, does not fall short of (resp. does not exceed) her threshold, i.e., $x^{\N}_{-i} \geq \tau(i)$ (resp. $\leq \tau(i)$). The superscript $\N$ in $x^\N_{-i}$ denotes the population size.

At each time index $k$, exactly one agent chosen uniformly at random becomes active and receives the chance to switch to her {{preferred}} strategy.
At time index $k+1$, the strategy of agent $i$ active at time index $k$ will be
\begin{align}
    &\mathtt{s}_{i}(k+1)  
    = \begin{cases}
    \1& \text{if } x^{\N}_{-i}(k) \geq {\tau(i)},   \\
    \2
    & \text{otherwise},
    \end{cases}
    \label{eq: scor}
\end{align}
if she is a coordinator and otherwise, if she is an anticoordinator,
\begin{align}
    &\mathtt{s}_{i}(k+1)  
    = \begin{cases}
    \1& \text{if } x^{\N}_{-i}(k) \leq {\tau(i)},   \\
    \2
    & \text{otherwise},
    \end{cases}
    \label{eq: santi}
\end{align}
where $\mathtt{s}_i$ is the strategy of agent $i$.

Coordinators (resp. anticoordinators) who have the same threshold build up a subpopulation, and there are altogether 
$\p'$ (resp. $\p$) subpopulations of coordinators (resp. anticoordinators).
Coordinating (resp. anticoordinating) subpopulations are labeled in the ascending (resp. descending) order of their thresholds by $1,\ldots,\p'$ (resp. $\p$), that is, $\tau'_1 < \tau'_2 < \ldots < \tau'_{\p'}$ (resp. $\tau_1 > \tau_2 > \ldots > \tau_{\p}$), where $\tau'_p$ (resp. $\tau_p$) is the threshold of the $p^{\text{th}}$ coordinating (resp. anticoordinating) subpopulation.  
So the heterogeneity of the population is captured by the distribution of population proportions over the total $\p + \p'$ subpopulations as
$
\bm{\rho} = ({\rho}_1, \ldots, {\rho}_{\p}, {\rho}'_{\p'}, \ldots, {\rho}'_1)^\top
$
where ${\rho}'_j$ (resp. ${\rho}_j$) represents the ratio of the number of agents in
coordinating (resp. anticoordinating)
subpopulation $j$  to the  population size $\N$.
For convenience of vector indexing, for $j\geq \p+1$, we define
 ${\rho}_j = {\rho}'_{\p + \p' +1-j}$. 
 Note that there is no subpopulation of size zero, meaning that $\min_j \{ \rho_j\} \geq \frac{1}{\N}$.
\begin{remark}  \label{remark_equivalenceofBestresponseAndThresholdDynamics}
    According to the best-response update rule, the active agent chooses the strategy which maximizes her current payoff, i.e, $s_i(k+1) = \1$ if $u_i^{\1}(k) \geq u_i^{\2}(k)$ and $\mathtt{s}_i(k+1) = \2$, otherwise where $u_i^{\mathtt{j}}(k)$ is the payoff to agent $i$ from choosing strategy $\mathtt{j} \in \{\1,\2\}$ at time index $k$.
    If the agents' payoffs are affine functions of the population proportion of strategy-$\1$ players,
    update rules \eqref{eq: scor} and \eqref{eq: santi} become equivalent to the best-response update rule. See \cite{ramazi2016networks} for more details.
\end{remark}

At the population level, the collective behavior of the agents in each subpopulation is of interest rather than that of each single agent \textit{per se}.
As a result, we define the \emph{population state} as the distribution of $\1$-players over the $(\p + \p')$ subpopulations at each index $k$ and denote it by
\begin{equation*}
    \x^{\N} 
    = (\mi{x}^{\N}_1,
    \ldots, 
    \mi{x}^{\N}_{{\p}},
    \mi{x}'^{\N}_{\p'},
    \ldots,
    \mi{x}'^{\N}_{{1}})^\top,
\end{equation*}
where  $\mi{x}'^{\N}_p$ (resp. $\mi{x}^{\N}_p$) represents the proportion of $\1$-players who belong to coordinating (resp. anticoordinating) subpopulation $p$.
We again use the notation
$x^{\N}_{p} = x'^{\N}_{\p + \p' + 1-p}$ for $p > \p$.
By defining the vector set
$$\bm{\mathcal{X}}_{ss} = \prod_{j=1}^{{\p}}[0,{\rho}_j] \times\prod_{j=1}^{{\p'}}[0,\rho'_{\p'-j+1}],$$
the resulting state space then equals 
$\bm{\mathcal{X}}_{ss} \cap {\frac{1}{\N}}\mathbb{Z}^{{\p + \p'}}$.

The evolution of the state defines the \emph{population dynamics} which are governed by update rules \eqref{eq: scor}, \eqref{eq: santi}, and the activation sequence of the agents. 
More specifically, the activation sequence is generated by the sequence of  mutually independent random variables $\langle {A}_k\rangle{_{k=0}^\infty}$ where ${A}_k$ is the agent active at time index $k$ and takes values in $[\N]$ with the distribution
$\mathbb{P}[{A}_k = i] =
    \frac{1}{\N}$. 
A realization of  an activation sequence $\langle {A}_k \rangle_{k=0}^{\infty}$ is \emph{persistent} \cite{ramazi2016networks}, if for every agent $i$ and every $K>0$ there exists a time index $T>K$ such that $A_T=i$.

As shown in \Cref{lem:discretedynamics}, the above definition of population dynamics is equivalent to the following compact version.
\begin{definition} \label{def:discreteDynamics}
    The \textbf{\emph{discrete {best-response} population dynamics}} are defined by the following discrete-time stochastic equation for $k=0,1,\ldots$
    \begin{align}\label{populationDynamicsDiscrete}
        \x^{\N}(k+1)  \!&=\! \x^{\N}(k) \!+\! \frac{1}{\N} \bigg(\!{S}_k-{s^*}\big({{P}_k},\x^{\N}{(k)}, S_k \big)\!\bigg)\mb{e}_{{P}_k},
    \end{align}
       where ${P}_k$ and ${S}_k$ are scalar random variables with distributions $\mathbb{P}[{P}_k  = p] = \rho_p$, $\mathbb{P} [{S}_k  = 1 \vert {P}_k  = p] = x^{\N}_p/\rho_p$,  {$\mathbb{P} [{S}_k  = 2 \vert {P}_k  = p] = 1 -x^{\N}_p/\rho_p$}, for $p \in [\p + \p']$, 
       $\mb{e}_{{P}_k}$ is the 
       ${P}_k$th standard basis vector in $\mathbb{R}^{\p + \p'}$, 
        and
     ${s^*}(p,\x^{\N},s)$ is a function where  ${s^*}(p,\x^{\N},1)$ returns $1$ (resp. $2$) if $\1$ (resp. $\2$) is the preferred strategy of an  $\1$-playing agent and ${s^*}(p,\x^{\N},2)$ returns $1$ (resp. $2$) if $\1$ (resp. $\2$) is the preferred strategy of a $\2$-playing agent in subpopulation $p$ at population state $\x^{\N}$ :
    \begin{alignat*}{2}
    {s^*}(&p,\x^{\N}, S_k) \\
    =& 
    \begin{cases}
    {1}& \hspace{-5pt}\text{if } (x^{\N} \leq {\tau_p}  \text{ and } p \leq \p \text{ and } S_k=2 )  \\
    & \text{ or }
       (x^{\N} \geq \tau_{\p + \p' +1 -p}'  \text{ and } p > \p \text{ and } S_k=2)    \\
        & \text{ or } (x^{\N}  \leq {\tau_p} + \frac{1}{\N} \text{ and } p \leq \p \text{ and } S_k=1 )  \\
    & \text{ or }
       (x^{\N}  \geq \tau_{\p + \p' +1 -p}'+ \frac{1}{\N} \text{ and } p > \p \text{ and } S_k=1) ,    \\
    {2}&\hspace{-5pt}\text{if } (x^{\N} > {\tau_p} \ \text{ and }  p \leq \p \text{ and } S_k=2)
     \\
     & \text{ or }
       (x^{\N}  < \tau_{\p + \p' +1 -p}'   \text{ and }  p > \p \text{ and } S_k=2)\\
    &\text{ or } (x^{\N} > {\tau_p}+ \frac{1}{\N} \ \text{ and }  p \leq \p \text{ and } S_k=1)
     \\
     & \text{ or }
       (x^{\N}  < \tau_{\p + \p' +1 -p}'+ \frac{1}{\N}  \text{ and }  p > \p \text{ and } S_k=1),
    \end{cases}
\end{alignat*}
where $x^\N = \sum_{p=1}^{\p + \p'}x_p^\N$.
\end{definition}
    In \Cref{def:discreteDynamics}, 
    the  random variable ${P}_k$ is the subpopulation of the agent active at time index $k$, and the  random variable  ${S}_k$ equals $1$ (resp. $2$) if the strategy of the active agent at index $k$ equals $\1$ (resp. $\2$). 
    \emph{Persistent} realizations of the sequences of random variables $\langle P_k \rangle_{k=0}^\infty$ and $\langle S_k\rangle_{k=0}^\infty$ refer to their realizations 
    under a persistent activation sequence.
    Function $s^*(p,\x^\N,\mathtt{j})$ returns  $1$ if  $\1$ is the preferred strategy of an agent playing $\mathtt{j}$ in the subpopulation $p$ when the population is at state $\x^{\N}$,  and it returns $2$ otherwise.
\begin{example} \label{exampleVaccine}
    In the context of public health, when a newly developed vaccine is introduced, some individuals choose to get vaccinated (i.e., strategy $\1$) as long as the vaccination coverage is below a certain threshold \cite{ibuka2014free}.
    For them, this signals that not enough individuals are vaccinated, and they want to be immunized.
    Conversely, some individuals go for vaccination when a specific proportion of the population has already been vaccinated \cite{SEANEHIA20172676} as then they perceive the vaccine as safe or a societal norm.
    The first group can be thought of as anticoordinators, while the latter are coordinators.
    The perception of what proportion of the population needs to be vaccinated to be considered rare or common varies among individuals, resulting in a heterogeneous population.
\end{example}
What is the asymptotic behavior of the population dynamics as the population size approaches infinity?
Does the population state reach an \emph{equilibrium point} where no agent tends to switch her current strategy? 
Or do the population dynamics undergo {\emph{perpetual fluctuations}?}

We define a fluctuation set as a non-empty subset of the state space that is (i) positively invariant, meaning that once the state enters the set it never leaves it under any activation sequence, (ii) it is minimal, meaning that it does not admit a positively invariant proper subset, implying that once the state enters it, every state in the set will be visited under a persistent activation sequence, and (iii) it is non-singleton, implying that it is not an equilibrium. 
In what follows, we define it in terms of the dynamics in \Cref{def:discreteDynamics}.
\begin{definition} 
A \textbf{\emph{fluctuation set}} for the dynamics  \eqref{populationDynamicsDiscrete} is a non-singleton set $\bm{\mathcal{Y}}^{\N}\subseteq \bm{\X}_{ss} \cap \frac{1}{\N}\mathbb{Z}^{\p + \p'}$ such that if $\x^{\N}(k)\in\bm{\mathcal{Y}}^{\N}$, for some $k$, then 
\emph{(i)} $\x^{\N}(k')\in\bm{\mathcal{Y}}^{\N}$ for all $k'>k$ and 
\emph{(ii)}
for every state $\bm{y}^{\N}\in\bm{\mathcal{Y}}^{\N}$ and every $K>0$ and under every persistent realizations of  $\langle P_k \rangle_{k=0}^\infty$ and $\langle S_k\rangle_{k=0}^\infty$ , 
 there exists some $T>K$ such that $\x^\N(k+T)=\bm{y}^{\N}$.
  The \emph{amplitude} of the perpetual fluctuation set $\bm{\mathcal{Y}}^{\N}$ is defined as 
    $
     \max_{\bm{y}^{\N},\bm{z}^{\N}\in  \bm{\mathcal{Y}}^{\N}} \vert \bm 1^\top\bm{y}^{\N} -  \bm 1^\top\bm{z}^{\N}\vert
 .$
 Starting from a given initial condition and under a specified realization of $\langle P_k \rangle_{k=0}^\infty$ and $\langle S_k\rangle_{k=0}^\infty$, we say that the dynamics \eqref{populationDynamicsDiscrete} undergo \textbf{\emph{perpetual fluctuations}} if $\x^\N(k) \in \bm{\mathcal{Y}}^{\N}$ for some fluctuation set $\bm{\mathcal{Y}}^{\N}$ and finite $k\geq 0$.
 \end{definition}

 {The dynamics of} a finite \emph{mixed} population consisting of both coordinating and anticoordinating subpopulations may {undergo perpetual fluctuations}
\cite{roohi}{,} the characterization of which appears to be challenging and remains unsolved. 
Could the asymptotic behavior of {discrete} populations be revealed by investigating the associated mean dynamics, as the population size approaches infinity?
\begin{figure}
    \centering
    \includegraphics[width=0.8\linewidth]{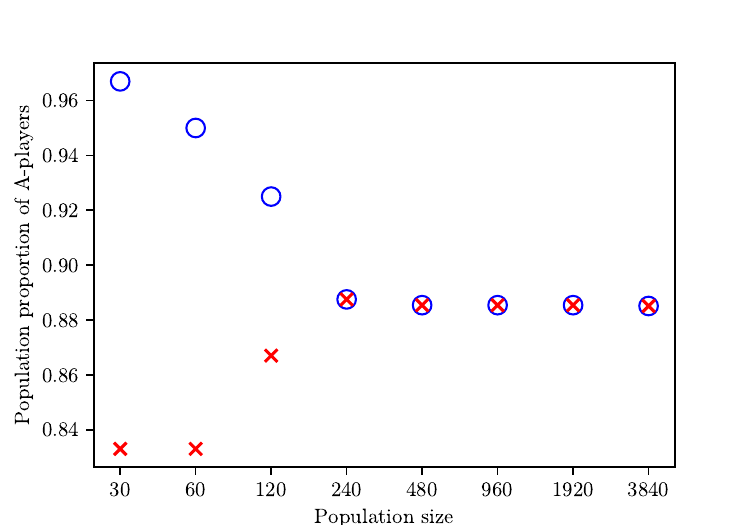}
    \caption{\textbf{The long-term fluctuations of the population proportion of strategy-$\1$ players for varying population sizes.} 
    The circles and crosses represent the maximum and minimum of the population proportion of $\1$-players, respectively.
    As the population size increases, these values get closer.}
    \label{fig:finitepop1}
\end{figure}
\begin{example} \label{exampleFinite}
    Consider a population of $\N$ agents stratified into one anticoordinating and five coordinating subpopulations.
    The thresholds and population proportions of the subpopulations are, respectively, as follows:
    $\smash{(\tr{\tau_1},{\tau'_5}, {\tau'_4}, {\tau'_3}, {\tau'_2}, {\tau'_1}) = (\tr{0.885},{0.89},{0.604}, ,{0.481},{0.444}, {0.21})}$ $\smash{(\tr{\rho_1}, {\rho'_5}, {\rho'_4}, {\rho'_3}, {\rho'_2}, {\rho'_1})=(\tr{12/30}, {3/30},{3/30},{3/30},{1/30},{8/30}})$.
    For each value of $\N=30,60,120,\ldots,3840$, the population state was simulated 100 times with different random activation sequences and initial conditions. 
    {For each value of $\N$, we considered $30\times\N$ steps.}
    For each run, the maximum and minimum of the population proportion of strategy $\1$ players over the last $2\N$ steps were recorded.
    These values for each population size $\N$ are depicted in Figure \ref{fig:finitepop1}.
    As the population size grows,  the maximum and minimum of the proportion of $\1-$players get closer. 
\end{example}
Example \ref{exampleFinite} gives us an intuition that the amplitude of fluctuations in the population proportion of $\1$-players decreases as the population size grows.
Does this observation extend to every mixed population of coordinators and anticoordinators as the population size approaches infinity i.e., $\N \rightarrow \infty$?
 (Fig \ref{fig:diag}).
The available results in stochastic approximation theory provide the basis for investigating this question.

\begin{figure}
    \centering

  
    \tikzset {_94pyk7jdl/.code = {\pgfsetadditionalshadetransform{ \pgftransformshift{\pgfpoint{0 bp } { 0 bp }  }  \pgftransformrotate{0 }  \pgftransformscale{2 }  }}}
    \pgfdeclarehorizontalshading{_98kebzmkh}{150bp}{rgb(0bp)=(1,1,1);
    rgb(37.5bp)=(1,1,1);
    rgb(58.73214176722935bp)=(0.87,0.95,0.98);
    rgb(62.5bp)=(0.63,0.85,0.94);
    rgb(100bp)=(0.63,0.85,0.94)}
    
      
    \tikzset {_rno3up0od/.code = {\pgfsetadditionalshadetransform{ \pgftransformshift{\pgfpoint{0 bp } { 0 bp }  }  \pgftransformrotate{0 }  \pgftransformscale{2 }  }}}
    \pgfdeclarehorizontalshading{_txet9lhh3}{150bp}{rgb(0bp)=(0.88,1,1);
    rgb(37.5bp)=(0.88,1,1);
    rgb(39.25bp)=(0.88,1,1);
    rgb(40.5bp)=(0.99,1,1);
    rgb(45bp)=(0.9,0.97,0.99);
    rgb(51bp)=(0.78,0.93,0.98);
    rgb(56.25bp)=(0.75,0.89,0.97);
    rgb(62.5bp)=(0.69,0.85,0.96);
    rgb(100bp)=(0.69,0.85,0.96)}
    \tikzset{every picture/.style={line width=0.75pt}} 
    
    \begin{tikzpicture}[x=0.75pt,y=0.75pt,yscale=-1,xscale=1]
    
    \path  [shading=_98kebzmkh,_94pyk7jdl] (114.5,2) .. controls (175.25,2) and (224.5,25.06) .. (224.5,53.5) .. controls (224.5,81.94) and (175.25,105) .. (114.5,105) .. controls (53.75,105) and (4.5,81.94) .. (4.5,53.5) .. controls (4.5,25.06) and (53.75,2) .. (114.5,2) -- cycle ; 
     \draw  [color={rgb, 255:red, 74; green, 144; blue, 226 }  ,draw opacity=1 ] (114.5,2) .. controls (175.25,2) and (224.5,25.06) .. (224.5,53.5) .. controls (224.5,81.94) and (175.25,105) .. (114.5,105) .. controls (53.75,105) and (4.5,81.94) .. (4.5,53.5) .. controls (4.5,25.06) and (53.75,2) .. (114.5,2) -- cycle ; 
    
    \path  [shading=_txet9lhh3,_rno3up0od] (114.5,125) .. controls (175.25,125) and (224.5,148.06) .. (224.5,176.5) .. controls (224.5,204.94) and (175.25,228) .. (114.5,228) .. controls (53.75,228) and (4.5,204.94) .. (4.5,176.5) .. controls (4.5,148.06) and (53.75,125) .. (114.5,125) -- cycle ; 
     \draw  [color={rgb, 255:red, 74; green, 144; blue, 226 }  ,draw opacity=1 ] (114.5,125) .. controls (175.25,125) and (224.5,148.06) .. (224.5,176.5) .. controls (224.5,204.94) and (175.25,228) .. (114.5,228) .. controls (53.75,228) and (4.5,204.94) .. (4.5,176.5) .. controls (4.5,148.06) and (53.75,125) .. (114.5,125) -- cycle ; 
    
    \draw  [color={rgb, 255:red, 74; green, 144; blue, 226 }  ,draw opacity=1 ][fill={rgb, 255:red, 74; green, 144; blue, 226 }  ,fill opacity=1 ] (79.5,117) -- (97,117) -- (97,105) -- (132,105) -- (132,117) -- (149.5,117) -- (114.5,125) -- cycle ;
    
    \draw (121,45.48) node [anchor=north west][inner sep=0.75pt]  [font=\footnotesize] [align=center] {{\fontfamily{ptm}\selectfont Equilibrates \cite{ramazi2020convergence}}\\{\fontfamily{ptm}\selectfont Equilibrates \cite{ramazi2017asynchronous}}\\{\fontfamily{ptm}\selectfont May Fluctuate \cite{roohi}}};
    \draw (240.88,23.46) node [anchor=north west][inner sep=0.75pt]  [font=\footnotesize,rotate=0] [align=center] {\begin{minipage}[lt]{44.79pt}\setlength\topsep{0pt}
     \begin{center}
    \centering
    {\fontfamily{ptm}\selectfont Discrete}
    {\fontfamily{ptm}\selectfont population}
    {\fontfamily{ptm}\selectfont of finite size $\N$}
    \end{center}
    \end{minipage}};
    \draw (39,45.58) node [anchor=north west][inner sep=0.75pt]  [font=\footnotesize] [align=center] {{\fontfamily{ptm}\selectfont Coordinating }\\{\fontfamily{ptm}\selectfont Anticoordinating}\\{\fontfamily{ptm}\selectfont Mixture of both} };
    \draw (142,170) node [anchor=north west][inner sep=0.75pt]  [font=\huge] [align=center] {{\huge {\fontfamily{ptm}\selectfont ?}}};
    \draw (240.88,149.07) node [anchor=north west][inner sep=0.75pt]  [font=\footnotesize,rotate=0]  {\begin{minipage}[lt]{44.79pt}\setlength\topsep{0pt}
    \begin{center}
    \centering
    {\fontfamily{ptm}\selectfont Discrete}
    {\fontfamily{ptm}\selectfont population}
    {\fontfamily{ptm}\selectfont $\N \rightarrow \infty$}
    \end{center}
    \end{minipage}};
    \draw (39,166.03) node [anchor=north west][inner sep=0.75pt]  [font=\footnotesize] [align=center] {{\fontfamily{ptm}\selectfont Coordinating }\\{\fontfamily{ptm}\selectfont Anticoordinating}\\{\fontfamily{ptm}\selectfont Mixture of both} };
    \draw (130.07,25.58) node [anchor=north west][inner sep=0.75pt]  [rotate=-0.53]  {$t\ \rightarrow \ \infty $};
    \draw (130.07,146.03) node [anchor=north west][inner sep=0.75pt]  [rotate=-0.53]  {$t\ \rightarrow \ \infty $};
    \draw (43.91,25.58) node [anchor=north west][inner sep=0.75pt]  [font=\footnotesize,rotate=-0.76] [align=center] {\begin{minipage}[lt]{40.13pt}\setlength\topsep{0pt}
    \begin{center}
    {\fontfamily{ptm}\selectfont  \textbf{Population}}
    \end{center}
    
    \end{minipage}};
    \draw (43.91,146.03) node [anchor=north west][inner sep=0.75pt]  [font=\footnotesize,rotate=-0.76] [align=center] {\begin{minipage}[lt]{40.13pt}\setlength\topsep{0pt}
    \begin{center}
    {\fontfamily{ptm}\selectfont  \textbf{Population}}
    \end{center}
    
    \end{minipage}};

    \end{tikzpicture}

    \caption{\textbf{How does the asymptotic behavior of the finite population of interacting agents \eqref{populationDynamicsDiscrete} evolve as the population size approaches infinity?}
    A finite population of all coordinators and a finite population of all anticoordinators each equilibrates in the long run,  \cite{ramazi2020convergence, ramazi2017asynchronous}. In the latter case, for some specific configurations, the population may admit a two-state cycle.
    The proportion of $\1$-players in a mixed finite population may undergo perpetual fluctuations
    \cite{roohi}. 
    }
    \label{fig:diag}
\end{figure}
\section{Background}

We present the following definitions laying the basis for introducing Theorem \ref{thm:implicationOFSandholm}, which connects the asymptotic behavior of the finite populations and their associated mean dynamics. 
Please see the appendix for more details.

A \emph{differential inclusion} is defined by 
\begin{equation} \label{eq_differentialInclusion}
    \dot{\x}(t) \in \bm{\mathcal{V}}\big(\x(t)\big),
\end{equation}
where $\bm{\mathcal{V}}$ is a set-valued map specifying a set of evolutions in $\mathbb{R}^n$ at each point $\x \in \mathbb{R}^n$. 
A \emph{Caratheodory solution} of \eqref{eq_differentialInclusion} on $[0,t_1]\subset \mbb{R}_{\geq 0}$ is an absolutely continuous map $\x:[0,t_1]\rightarrow \mathbb{R}^n $ such that for almost all $t \in [0,t_1]$,  $\dot{\x}(t) \in \bm{\mathcal{V}}(\x(t))$ \cite{cortes2008discontinuous}. 
{Robotic manipulators interacting with objects, high-speed supercavitation vehicles, and power systems can be well modeled and analyzed within the framework of differential inclusions \cite{cortes2008discontinuous, kani, powerSystem}.}
A point $\x^* \in \mathbb{R}^n$ that satisfies ${\mb{0}} \in \bm{\mathcal{V}}(\x^*)$ is an \emph{equilibrium} of \eqref{eq_differentialInclusion}  \cite{cortes2008discontinuous}.
The \emph{set-valued dynamical system} induced by the differential inclusion \eqref{eq_differentialInclusion} is denoted by $\bm{\Phi}(\x_0)$ and defined as the set of all solutions of \eqref{eq_differentialInclusion} with initial condition $\x_0$ \cite{roth2013stochastic}.
A set-valued map $\bm{\mathcal{V}}$ is called \emph{upper semicontinuous} if its graph, $\{(\x,\y) \vert \y \in \bm{\mathcal{V}}(\x) \}$, is closed.
 When an upper semicontinuous map is equipped with three additional properties of nonemptiness, convexity, and boundedness, the map is \emph{good upper semicontinuous}.
\begin{definition}[\cite{roth2013stochastic}]  \label{def_Brikhoff}
    Let $\bm{\Phi}$ be the dynamical system induced by \eqref{eq_differentialInclusion}.
     The recurrent points of $\bm{\Phi}$ are defined as
        $$  \bm{\mathcal{R}}_{\bm{\Phi}} 
        = \left \{\x_0 \Big \lvert 
        \x_0 \in \bm{\mathcal{L}}(\x_0) \right \},
    $$
    where $\bm{\mathcal{L}}(\x_0)$ is the \textbf{\emph{limit set}} of point  $\x_0$ defined by $\bigcup_{\y \in \bm{\mathcal{S}}_{\x_0}}  \bigcap_{t \geq 0} \text{cl}\left(\y[t, \infty\right))$ and $\bm{\mathcal{S}}_{\x_0}$ is the set of solutions of differential inclusion \eqref{eq_differentialInclusion} with the initial condition $\x_0$.
    The \textbf{\emph{Birkhoff center}} of $\bm{\Phi}$, $\mathrm{BC}(\bm{\Phi})$, is defined as the closure of $\bm{\mathcal{R}}_{\bm{\Phi}}$, i.e., $\mathrm{BC}(\bm{\Phi}) = \text{cl}(\bm{\mathcal{R}}_{\bm{\Phi}})$.
\end{definition}
Each $\y$ in \Cref{def_Brikhoff} is a solution of the differential inclusion, and the term $\bigcap_{t \geq 0} \text{cl}(\y[t, \infty))$ is the limit set of $\y$.
In words, $\x_0$ is a recurrent point of $\bm{\Phi}$ if it is included in the limit set of some solution that starts from $\x_0$.
Namely, a recurrent point generalizes the notion of an equilibrium point, as, for example, the points on a limit cycle are recurrent.

How to approximate the realizations of a collection of Markov chains $\langle \mb{X}^{\epsilon}_k\rangle{_{k=0}^\infty}$ for a sequence of $\epsilon>0$ with solutions of a differential inclusion $\dot{\x} \in \bm{\mathcal{V}}(\x)$? 
The following definition establishes the basis for doing so.
\begin{definition} \label{defGSAP} \cite{roth2013stochastic}
    Consider the good upper semicontinuous differential inclusion $\dot{\x} \in \bm{\mathcal{V}}(\x)$ over the convex and compact state space $\bm{\mathcal{X}}_0$.
    For a sequence of positive values of scalar $\epsilon$ approaching $0$, let $\mathbi{U}^{\epsilon} = \langle \mathbi{U}^{\epsilon}_k\rangle{_{k=0}^\infty}$ be a sequence of $\mathbb{R}^n$-valued random variables and $\langle\bm{\mathcal{V}}^{\epsilon}\rangle$ be  a family of set-valued maps on $\mathbb{R}^n$. 
    We say that $\langle \langle \mb{X}^{\epsilon}_k\rangle{_{k=0}^\infty}\rangle_{\epsilon >0}$ is a family of \emph{generalized stochastic approximation processes} (or a \emph{GSAP}) for the differential inclusion $\dot{\x} \in \bm{\mathcal{V}}(\x)$ if the following conditions are satisfied:
    \begin{enumerate}
        \item $\mb{X}^{\epsilon}_k \in {\bm{\mathcal{X}}_0}$ for all $k \geq 0$, 
        \item $\mb{X}^{\epsilon}_{k+1} - \mb{X}^{\epsilon}_k - \epsilon \mathbi{U}^{\epsilon}_{k+1} \in \epsilon \bm{\mathcal{V}}^{\epsilon}(\mb{X}^{\epsilon}_k)$,
        \item for any $\delta >0$, there exists an $\epsilon_0 > 0$ such that for all $\epsilon \leq \epsilon_0$ and $\x \in \bm{\mathcal{X}}_0$, 
        \begin{equation} \label{eq_thirdCondition}
        \begin{aligned}
        \bm{\mathcal{V}}^{\epsilon}(\x) \subset  \{\bm{z} &\in \mathbb{R}^n \mid \exists \y:\vert\x-\y\vert< \delta, 
        \\& \inf_{\bm{v}\in \bm{\mathcal{V}}(\y)} \!\vert \bm{z} - \bm{v}\vert< \delta \},
        \end{aligned}
        \end{equation}
        \item for all $T>0$ and for all $\alpha >0$, 
        $$
        \lim_{\epsilon \rightarrow 0} \mathbb{P} 
        \left[ \max_{k \leq \frac{T}{\epsilon}} \left \vert 
        \textstyle\sum_{i = 1}^{k} \epsilon \mathbi{U}^{\epsilon}_{i} \right \vert > \alpha \mid \mb{X}^{\epsilon}_0 = \x \right] = 0
        $$
        uniformly in $\x \in \bm{\mathcal{X}}_0$.
    \end{enumerate}
\end{definition}
The GSAPs can be regarded as perturbed solutions of the differential inclusions.
The first condition ensures that the perturbed solutions remain in the same set as the unperturbed ones.
The second {condition defines $\langle \mb{X}^{\epsilon}_{k} \rangle{_{k=0}^\infty}$ as a  sequence of elements which are recursively updated.
At each index $k$, the value of $\mb{X}^{\epsilon}_{k+1}$ is characterized based on those of  $\mb{X}^{\epsilon}_{k}$, $\bm{\mathcal{V}}^{\epsilon}(\mb{X}^{\epsilon}_k)$,  and the random variable  $\mathbi{U}^{\epsilon}_{k+1}$.}
The third condition ensures that the collection $\langle\bm{\mathcal{V}}^{\epsilon}\rangle$ approaches to the vicinity of $\bm{\mathcal{V}}$ as $\epsilon$ approaches zero.
{More specifically, the third condition is indeed a limit condition, which requires that for an arbitrarily small $\delta>0$  there exists $\epsilon_0>0$  such that for each $\epsilon < \epsilon_0$ in the sequence $\langle \epsilon \rangle$, the values belonging to the set valued map $\bm{\mathcal{V}}^\epsilon(\x)$ are within a $\delta$-neighborhood of the image of the $\delta$-neighborhood of $\x$.}
Well-behaved random variables, such as sub-Gaussian random variables and some classes of random variables with finite higher moments, satisfy the last condition. 

In \cite{roth2013stochastic}, Roth and Sandholm
studied the limiting stationary distributions of these GSAPs and showed that if for each $\epsilon >0$,  
 $ \langle \mb{X}^{\epsilon}_{k}\rangle$ is a Markov chain,
 their stationary measures as $\epsilon$ approaches zero will be concentrated on $\mathrm{BC} (\bm{\Phi})$.
{The following theorem is a result of \Cref{thm:sandholm} and was implicitly used in the proof of \cite[Theorem 4.11]{roth2013stochastic}.}

\begin{theorem} \label{thm:implicationOFSandholm}
For a vanishing sequence  $\langle \epsilon_n \rangle_{n=0}^{\infty}$ let
    $ \langle \langle \mb{X}^{\epsilon_n}_k\rangle{_{k=0}^\infty}\rangle_{n}$
    be GSAPs for a good upper semicontinuous differential inclusion $\dot{\x} \in \bm{\mathcal{V}}(\x)$.
   Assume that for each $\epsilon_n$, $\langle \mb{X}^{\epsilon_n}_k\rangle{_{k=0}^\infty}$
    is a Markov chain and ${\mu}^{\epsilon_n}$ is an invariant probability measure of $\langle \mb{X}^{\epsilon_n}_k\rangle{_{k=0}^\infty}$.
    We then have  $\lim_{n\to \infty} {\mu}^{\epsilon_n}(\bm{\mathcal{O}}) = 1$, where $\bm{\mathcal{O}}$ is any open set containing the Birkhoff center of the dynamical system $\bm{\Phi}$  induced by $\dot{\x} \in \bm{\mathcal{V}}(\x)$.
\end{theorem}
\section{The link to the continuous-time dynamics}
Our aim is to determine the asymptotic behavior of the population dynamics \eqref{populationDynamicsDiscrete} as the population size $\N$ approaches infinity.
Theorem \ref{thm:implicationOFSandholm} connects the asymptotic behavior of the GSAPs to the steady-state behavior of their associated differential inclusion.
Hence, if one shows that, firstly, dynamics \eqref{populationDynamicsDiscrete} define a Markov chain, and
secondly, the collection of the Markov chains indexed by the population size is a GSAP for a good upper semicontinuous differential inclusion,
then the results of Theorem \ref{thm:implicationOFSandholm} can be applied, meaning that the asymptotic behavior of the population dynamics can be revealed by analyzing the differential inclusion. 

In view of \eqref{populationDynamicsDiscrete}, the next state of the population dynamics only depends on the current population state, the strategy of and the subpopulation associated with the current active agent.
As a result, the sequence of $\langle\x^{\N}(k)\rangle{_{k=0}^\infty}$ is a realization of a Markov chain.
\begin{prop} \label{propMarkov}
    The sequence $\langle \x^{\N}(k)\rangle{{_{k=0}^\infty}}$ is a realization of {the Markov chain} $\langle\mathbf{X}^{\frac{1}{\N}}_k\rangle{_{k=0}^\infty}$
    with the state space $\bm{\mathcal{X}}_{ss} \cap \frac{1}{\N}\mathbb{Z}^{{\p + \p'}}$, initial state $\mathbf{X}^{\frac{1}{\N}}_0 = \x^{\N}(0)$
    and transition probabilities
   \begin{alignat}{2}
         &\text{Pr}_{\x^{\N},\y^{\N}} = \nonumber\\  
    &\begin{cases}
        ({\rho}_p - x^{\N}_p)(2-{s^*}(p,\x^{\N}, 2)) &\hspace{-35pt}\text{if }\y^{\N} = \frac{1}{\N}\mb{e}_p + \x^{\N},  \\
        x^{\N}_p ({s^*}(p,\x^{\N}, 1)-1) &\hspace{-35pt}\text{if } \y^{\N} = -\frac{1}{\N}\mb{e}_p + \x^{\N},
    \\
        1-\displaystyle\sum_{p=1}^{\p + \p'} \big(  x^{\N}_p({s^*}(p,\x^{\N},1)-1) \hspace{-2210pt}&\hspace{-35pt}\text{if } \y^{\N} =  \x^{\N},\\ \hspace{27pt} + 
        (\rho_p-x^{\N}_p) (2-{s^*}(p,\x^{\N},2)) 
        \big) 
         \\
        0 &\hspace{-35pt}\text{otherwise}, 
    \end{cases} \label{eq:markov}
\end{alignat} 
where $\mb{e}_p$ is the $p^{\text{th}}$ standard basis vector in $\mathbb{R}^{(\p + \p')}$.
\end{prop}
See Appendix for a proof.
\begin{definition}
    The \textbf{\emph{population dynamics Markov chain}} is the Markov chain $\langle \mathbf{X}^{\frac{1}{\N}}_k\rangle{_{k=0}^\infty}$  with the state space $\bm{\mathcal{X}}_{ss} \cap \frac{1}{\N}\mathbb{Z}^{{\p + \p'}}$, initial state $\mathbf{X}^{\frac{1}{\N}}_0 = \x^{\N}(0),$
    and transition probabilities {\eqref{eq:markov}.}
\end{definition}
   We increase the population size $\N$ in a way that the population structure, i.e.,  the population proportions of the subpopulations (vector $\bm{\rho}$) remains unchanged.
    Hence, the sequence according to which the population size approaches infinity should satisfy $\N  \bm{\rho} \in \mathbb{Z}^{\p + \p'}$.   
    We denote the set of such valid population sizes by $\mathcal{N}$. 
 
Now, to apply Theorem \ref{thm:implicationOFSandholm}, we need to show that the collection of $\langle \langle \mb{X}^{\frac{1}
{\N}}_k\rangle{_{k=0}^\infty} \rangle_{\N \in \mathcal{N}}$
is a GSAP for a good upper semicontinuous differential inclusion, and hence the support of the limit point of its invariant probability measures is determined by the steady-state behavior of that differential inclusion.
We claim that $\langle \langle \mb{X}^{\frac{1}
{\N}}_k\rangle{_{k=0}^\infty} \rangle_{\N \in \mathcal{N}}$
is a GSAP for the following differential inclusion:
\begin{definition}  \label{def_semicontinuousDynamics}
    The \textbf{\emph{continuous-time} population dynamics} are defined by 
       $ \dot{\x}\in \bm{\mathcal{V}}(\x)$
    where $\bm{\mathcal{V}} :\bm{\mathcal{X}}_{ss} \rightrightarrows \mathbb{R}^{\p+\p'}$, 
    for $p \leq \p $
\begin{subequations} 
\begin{align}
    \hspace{-55pt} {\V}_p(\x) = \begin{cases} \label{222}
    \{{\rho}_p - x_p\} & \text{if }x < \tau_p,   \\
    [-x_p, {\rho}_p-x_p] & \text{if }  x = \tau_p, \\
    \{- x_p\} & \text{if } x > \tau_p,   
    \end{cases}
    \end{align} \label{eq: type_mixed}
\end{subequations}
and for  $p > \p$ 
\begin{align*}
    &\V_p(\x) = \begin{cases}
   \{ - x_p \} & \text{if } x < \tau'_{\p + \p'+1-p},  \\
    [-x_p, {\rho}_{p}-x_p] & \text{if } x = \tau'_{\p + \p'+1-p},\\
     \{ {\rho}_{p} - x_p \} & \text{if } x > \tau'_{\p + \p'+1-p}, 
    \end{cases}  \tag{\ref{eq: type_mixed}b}
    \end{align*}
where  $x=\sum_{p=1}^{\p + \p'}x_p$.
\end{definition}
The above dynamics are indeed the \emph{mean dynamics} of the discrete population dynamics \eqref{populationDynamicsDiscrete} \cite{sandholm2010population}.
The mean dynamic considers the decision-making process as a continuous process with an infinitely large population, and $\dot{x}_p$ describes the rate of change in the population proportion of strategy-$\1$ players in subpopulation $p$.
This equals the expected inflow, $F_{p,\1}$, that is, the rate of change to strategy $\1$ minus the expected outflow, $F_{p,\2}$, that is, the rate of change to strategy $\2$:
\begin{equation} \label{eq_meanDynamics}
    \dot{x}_p \in \{a -b \vert a \in F_{p,\1}(\x), b \in F_{p,\2}(\x)\},
\end{equation}
where $F_{p,\1},F_{p,\2}:\bm{\X}_{ss} \rightrightarrows {[0,\rho_p]}$.
If the population proportion of strategy-$\1$ players is greater (resp. less) than the threshold of anticoordinating (resp. coordinating) subpopulation $p$, then the inflow will be zero, and the outflow equals the current proportion of strategy-$\1$ players of subpopulation $p$, i.e., $\dot{x}_p = -x_p$.
On the contrary, if the population proportion of strategy-$\1$ players is less (resp. greater) than the threshold of anticoordinating (resp. coordinating) subpopulation $p$, then the outflow will be zero and the inflow equals the current proportion of strategy $\2$ players of subpopulation $p$, i.e., $\dot{x}_p = \rho_p-x_p$.
When the population proportion of strategy-$\1$ players equals the threshold of (anticoordinating or coordinating) subpopulation $p$, both strategies $\1$ and $\2$ are a legitimate choice should there be no tie-breaker, and some may choose strategy $\1$ while others may choose strategy $\2$.
Therefore, the rate will be bounded by the above two cases, i.e., $-x_p$ and $\rho_p-x_p$, resulting in the interval $[-x_p, \rho_p-x_p]$.
If there is a tie-breaker, such as the set-up in this paper, then one of the strategies is preferred, resulting in one of the above two cases, which results in a \emph{selection} from the defined differential inclusion.
A \emph{selection} from $\bm{\mathcal{V}}(\x)$ is a singleton map $\bm{\nu}(\x) \in \mathbb{R}^n$ which satisfies $\bm{\nu}(\x) \in \bm{\mathcal{V}}(\x)$ for all $\x$ in $\mathbb{R}^n$ \cite{smirnov2002introduction}.
That is why the mean dynamics of the discrete population dynamics \eqref{populationDynamicsDiscrete} is sometimes considered as a selection from \eqref{eq: type_mixed} \cite{roth2013stochastic}.
By extending the domain of the continuous-time population dynamics to $\mathbb{R}^{\p+\p'}$ using projection and in view of \cite[Theorem 6.A.2]{sandholm2010population}, the good upper semicontinuity of
the differential inclusion \eqref{eq: type_mixed} can be verified.

\begin{remark} \label{remarkBR}
    The equivalent equation to \eqref{eq_meanDynamics} in the context of best-response would be 
    \begin{equation*}
        \dot{\x} \in  \{ \y -\x \vert \y \in \bm{\mathrm{Br}}(\x)\}
    \end{equation*}
    where $\bm{\mathrm{Br}}:\bm{\X}_{ss}\rightrightarrows{\bm{\X}_{ss}}$ is the best-response correspondence where $\mathrm{Br}_p(\x)$ returns that population proportion of strategy-$\1$ players of subpopulation $p$ maximizing their utility defined in the game-theoretic context of the linear threshold models (see \Cref{remark_equivalenceofBestresponseAndThresholdDynamics}).
\end{remark}

\begin{lemma}    \label{lemGSAPS}
    Assume that the thresholds of the subpopulations are distinct.
    Then the collection of $\langle \langle \mb{X}^{\frac{1}{\N}}_k\rangle{_{k=0}^\infty}\rangle_{\N\in\mathcal{N}}$ is a family of GSAPs for \eqref{eq: type_mixed}.
\end{lemma}
For the sake of readability, the proofs of the results are provided in the appendix.
Based on Lemma \ref{lemGSAPS}, the population dynamics \eqref{populationDynamicsDiscrete} define GSAPs for the differential inclusion \eqref{eq: type_mixed}. 
Hence, in view of Theorem \ref{thm:implicationOFSandholm}, the support of the limit point of the stationary distributions of \eqref{eq:markov}, as $\N \rightarrow \infty$, is contained in the Birkhoff center of $\bm{\Phi}$ induced by \eqref{eq: type_mixed}.
In the next section, we determine the Birkhoff center of $\bm{\Phi}$.
\begin{remark} \label{setupRemark}
    The set-up in \eqref{eq: scor}, \eqref{eq: santi}, and \eqref{populationDynamicsDiscrete} is based on a tie-breaking rule \cite{sakhaei2022equilibration} where the agents prefer strategy $\1$ when the population proportion of $\1$-players, excluding themselves, equals the agent's threshold, i.e., $x^{\N}_{-i}=\tau(i)$ for every agent $i$.
    However, for other tie-breaking rules, such as when strategy $\2$  is preferred, or the agents uniformly randomly choose either of the strategies, Lemma \ref{lemGSAPS} remains valid (see Corollary A1).
    The reason is that a different tie-breaker will only change $s^*(\cdot, \cdot, \cdot)$ in \Cref{def:discreteDynamics} and consequently  
    the selection from \eqref{eq: type_mixed} when $x = \tau_p$ where $p$ is the type of agent $i$.
        Accordingly, the dynamical behavior of the resulting selection is already captured by \eqref{eq: type_mixed}.
    Therefore, investigating the asymptotic behavior of \eqref{eq: type_mixed} does reveal that of the discrete population dynamics \eqref{populationDynamicsDiscrete} with such tie-breaking rules as population size approaches infinity.
\end{remark}
\section{The analysis of the continuous-time dynamics} 
\label{sec:mixed}
To analyze the continuous-time dynamics, the equilibrium points of \eqref{eq: type_mixed} should be characterized.
The following subsection provides an intuition for how analyzing the evolution of $x(t)$ in \Cref{def_semicontinuousDynamics}, i.e., the population proportion of $\1$-players, helps that of continuous-time population dynamics.
\subsection{Intuition and example} \label{subsec: intuition and example}
\begin{example} \label{ex:abstract}
Consider a mixed population of one anticoordinating and two coordinating subpopulations.
        The thresholds and the population proportions of the subpopulations are respectively as follows: $(\tr{\tau_1}, {\tau_2'}, {\tau_1'}) = (\tr{0.85},{0.75},{0.35})$ and $(\tr{\rho_1}, {\rho_2'}, {\rho_1'}) = (\tr{0.6},{0.3},{0.1})$. 
    The continuous-time population dynamics are $\dot{\x} \in \bm{\mathcal{V}}(\x)$, where $\bm{\mathcal{V}} = [\V_1,\V_2,\V_3]^\top$ and
    \begin{subequations}
    \begin{equation}
     {\V}_1(\x) = \begin{cases}
    \{0.6 - x_1\} & \text{if } x < 0.85,   \\
    [-x_1, 0.6-x_1] & \text{if }  x = 0.85, \\
    \{- x_1\} & \text{if } x > 0.85, 
    \end{cases}
    \end{equation}\label{eq:example2}
    \end{subequations}
    \begin{equation}
    \V_2(\x) = \begin{cases} 
     \{ - x_2 \} & \text{if } x < 0.75,  \\
     [-x_2, 0.3-x_2] & \text{if}\quad  x = 0.75,\\
   \{ 0.3 - x_2 \} & \text{if }  x > 0.75, 
    \end{cases}\tag{\ref{eq:example2}b}
    \end{equation}
    \begin{equation}
     \V_3(\x) = \begin{cases} \nonumber
       \{ - x_3 \} & \text{if } x < 0.35, \nonumber \\
      [-x_3, 0.1-x_3] & \text{if}\quad  x = 0.35,\\
   \{ 0.1 - x_3 \} & \text{if } x > 0.35.\nonumber 
    \end{cases}\tag{\ref{eq:example2}c}
\end{equation}

\begin{itemize}
\item {Motivation.}
To find the equilibrium points, we should find the state populations $\x^*$ at which  $\bm{0} \in \bm{\mathcal{V}}(\x^*)$.
The combinations of all possible cases for $\V_1, \V_2$, and $\V_3$ should be considered based on the value of $x$.
So if we already knew the value of $x$ at equilibrium, we would only need to investigate the active case in $\V_1$, $\V_2$, and $\V_3$ to find $x_1,x_2,x_3$, and in turn $\bm{x}$ at equilibrium.
Hence, we obtain the dynamics $x(t)$ as follows.
\item {The abstract dynamics.}
The thresholds divide the unit interval into four disjoint open intervals, i.e., $(0,0.35)$, $(0.35, 0.75)$, $(0.75,0.85)$, and $(0.85,1)$. 
When $x \in (0,0.35)$, $\bm{\mathcal{V}}(\x)$ is a singleton and $\dot{\x}$ will be equal to $(0.6-x_1,-x_2,-x_3)$.
The derivative of $x$ will then be $\dot{x} =  0.6 - x$.
A similar procedure can be applied to the other intervals, resulting in $\dot{x} \in \mathcal{X}(x)$, where 
         \begin{gather} 
        \scalebox{0.94}{$
       \mathcal{X}(x) = \begin{cases}
       \{0.6-x\} \quad &\text{if }  x \in [0,0.35),\\
        [0.6-x, 0.7-x]  \quad&\text{if } x = 0.35,\\
        \{0.7-x\} \quad &\text{if }  x \in (0.35, 0.75),\\
        [0.7-x,1-x] \quad &\text{if }  x = 0.75,\\
        \{1-x\} \quad &\text{if }  x \in (0.75, 0.85),\\
        [0.4-x, 1-x] \quad&\text{if }  x = 0.85,\\
        \{0.4-x\} \quad &\text{if }  x \in (0.85, 1].
        \end{cases}
    $ } \label{eq: example}
    \end{gather}   
    \item {Abstract dynamics' equilibria.}
        At equilibrium, we have $0 \in \mathcal{X}(x^*)$.
        So we can find the equilibria by investigating all intervals in \eqref{eq: example}.
        The first two intervals are equilibrium-free.
        However, the image of the interval $(0.35,0.75)$ under $\X$ is $(-0.05,0.35)$, which includes zero. 
        Therefore, the abstract dynamics admit an equilibrium point in this interval, which turns out to be $x^*_1 = 0.7$ as $0 \in \mathcal{X}(0.7)$.
        Following the same procedure, we find two additional equilibrium points $x^*_2 = 0.75$ and $x^*_3 = 0.85$ for the abstract dynamics.
        \item {Population dynamics' equilibria.}
        In view of \eqref{eq:example2}, the value $x^*_1$ equals the sum of $\rho_1$ and $\rho'_1$, i.e., the population proportions of those subpopulations whose {preferred} strategy at $x^*_1$ is $\1$.
        Evaluating \eqref{eq:example2} at $x = x^*_1$, we see that  $\x^*_1 =(0.6,0,0.1)$ is an equilibrium point for the population dynamics.
        The point $x^*_2$ equals the threshold of the second coordinating subpopulation, and the population state  $\x^*_2 =(0.6,0.05,0.1)$ is an equilibrium point for \eqref{eq:example2}.
        At $\x^*_2$, except for the second coordinating subpopulation, members of the same subpopulations adopt the same strategies, i.e., $x^*_{2,1} = \rho_1$, $x^*_{2,3} = \rho'_1$.
        Finally, the point $x^*_3 = 0.85$ 
        equals the anticoordinating subpopulation's threshold, and  $\x^*_3 =(0.45,0.3,0.1)$ is an equilibrium for the population dynamics.
        At $\x^*_3$, except for the anticoordinators, members of the same subpopulations adopt the same strategies, i.e., $x^*_{3,2} = \rho'_2$, $x^*_{3,3} = \rho'_1$.
        \end{itemize}
\end{example}
    \vspace{-10pt}
\subsection{The abstract dynamics.}
\Cref{ex:abstract} gives us an intuition that 
it is useful to focus on the evolution of the population proportion of $\1$-players, i.e., $x(t)$, which from now on we refer to as the \emph{abstract state} as it is abstracting away the heterogeneity of the population. 
Define ${\pi}'_j$ (resp. ${\pi}_i$) as the cumulative population proportion of coordinating (resp. anticoordinating) subpopulations, that is, the population proportion of those having thresholds equal to or less (resp. greater) than $\tau'_j$ (resp.  $\tau_i$), i.e., 
\begin{equation*}
   {\pi}'_j = \sum_{k=1}^j \rho'_k, \quad
   {\pi}_i = \sum_{k=1}^i \rho_k, 
\end{equation*}
where $j \in [\p']$, $i \in [\p]$,
and we define ${\pi}'_0 = {\pi}_0 =  {\pi}_{-1} = {\pi}'_{-1} =  0$,
${\pi}'_{\p'+1} = {\pi}'_{\p'},{\pi}_{\p+1} = {\pi}_{\p}$,
$\tau_{\p+1} = \tau'_0 = 0$, and $\tau_0 = \tau'_{\p'+1} = 1$.
  \begin{assumption} \label{ass:ass1}
The thresholds of the subpopulations are unique and satisfy the following:
    $$\forall k \!\in [\p]\cup\{0\}\forall l \in [\p']\cup\{0\} ({\pi}_k + {\pi}'_l \notin \! \{\tau_1,\ldots, \tau_{\p}, \tau'_1,\ldots, \tau'_{\p'} \}).$$  
\end{assumption}
In words, it is assumed that the sum of the cumulative population proportion of the first $k$ anticoordinating subpopulations and first $l$ coordinating subpopulations does not match any of the thresholds for the $(\p + \p')$ subpopulations.
\begin{prop} \label{prop2}
    Consider the continuous-time dynamics \eqref{eq: type_mixed}.
    Under \Cref{ass:ass1}, the evolution of the population proportion of $\1$-players $x(t)$ is governed by 
    \begin{gather} 
    \dot{x} \in \mathcal{X}(x), \label{eq: abstract_het}\\
\scalebox{0.94}{$
   \mathcal{X}(x) = \begin{cases}
   [{\pi}_{i-1} + {\pi}'_{j}-x, {\pi}_{i} + {\pi}'_{j}-x]  \hspace{50pt}\text{if }  \exists i (x = \tau_{i}),&
    \\
     [{\pi}_{i} + {\pi}'_{j-1}-x, {\pi}_{i} + {\pi}'_{j}-x] \hspace{50pt} \text{if }  \exists j (x = \tau'_{j}),&\\
     \{{\pi}_{i} + {\pi}'_{j} -x\}  \\ \quad \qquad   \text{if }  \exists i,j (x \in   (\max \{\tau'_{j}, \tau_{i+1}\},\min \{\tau'_{j+1}, \tau_{i} \})),
    \end{cases}
   $ } \nonumber
\end{gather}
where for $x = \tau_i$, $j=\max \{k \in [\p'] \cup \{0\}| \tau'_k < \tau_{i} \} $  and  for $x = \tau_j'$,  $i$ is equal to $\max \{k \in [\p] \cup \{0\}| \tau'_j < \tau_{k} \} $.
\end{prop}
We refer to \eqref{eq: abstract_het} as the \emph{abstract dynamics}.
\subsection{Equilibrium points}
     Inspired by the definitions in \cite{ramazi2017asynchronous} and based on \Cref{ex:abstract}, we claim that the equilibria of the population dynamics are either \emph{clean-cut} or \emph{ruffled} as defined in the following.
\begin{definition}
    A \textbf{\emph{clean-cut population state}}, {is a population state $\mb{c}^{ij}, i \in [\p],j \in [\p'],$ } defined by 
    $$\mb{c}^{ij} = ({\rho}_1,\ldots, {\rho}_{i},0,\ldots, 0,{\rho}'_{j}, \ldots ,{\rho}'_1).$$
    The abstract state $x$ at $\mb{c}^{ij}$ is denoted by ${c}^{ij} = {\pi}_i + {\pi}'_j$. 
    The state at which no coordinating (resp. anticoordinating) agents adopt strategy $\1$, is denoted by $\mb{c}^{i0}$ (resp. $\mb{c}^{0j}$).
\end{definition}
    In a clean-cut population state, all members of every subpopulation adopt the same strategy
    in a way that all coordinating (resp. anticoordinating) subpopulations with thresholds equal to or less (resp. greater) than that of some benchmark coordinating (resp. anticoordinating) subpopulation $j \in [\p']$ (resp. $i \in [\p]$) play strategy $\1$ and the remaining subpopulations play strategy $\2$.  
\begin{definition}
    A {\textbf{\emph{ruffled anticoordinator-driven population state}}} is a population  state $\mb{a}^{ij}, i \in [\p],j \in [\p'],$  defined by
     \begin{equation*}
    \mb{a}^{ij} = 
    \big({\rho}_1,\ldots, {\rho}_{i-1},\tau_{i} - ({\pi}'_{j}+ {\pi}_{i-1}), 0,
    \ldots,0 ,{\rho}'_{j}, \ldots, {\rho}'_{1}\big),
    \end{equation*}
     where $j$ equals $\max \{k \in [\p'] \cup \{0\}| \tau'_k < \tau_{i} \} $.
    The abstract state at $\mb{a}^{ij}$ is denoted by $a^{ij}=\tau_{i}$.
    For $j=0$, no coordinators play $\1$.
\end{definition}
In a ruffled anticoordinator-driven population state, the population proportion of strategy-$\1$ players equals the threshold of some benchmark anticoordinating subpopulation $i$ such that all coordinating  (resp. anticoordinating) subpopulations with a threshold less (resp. greater) than that of subpopulation $i$ play strategy $\1$.
\begin{definition}
    A {\textbf{\emph{{ruffled coordinator-driven} population state}}} is a population state $\mb{o}^{ij}, i \in [\p], j\in [\p']$, defined by
     \begin{align*}
    \mb{o}^{ij} = \big({\rho}_1,\ldots, {\rho}_{i}, 0,\ldots,0,& \tau'_{j} - ({\pi}'_{j-1} + {\pi}_{i}),{\rho}'_{j-1}\ldots, {\rho}_1'\big), 
    \end{align*}
      where  $i= \max \{k \in [\p] \cup \{0\}| \tau'_j < \tau_{k} \} $.
   The abstract state at $\mb{o}^{ij}$ is denoted by $o^{ij}=\tau'_{j}$.
   For $i=0$, no anticoordinators play $\1$.
\end{definition}
 In a ruffled coordinator-driven population state, the population proportion of strategy-$\1$ players equals the threshold of some benchmark coordinating subpopulation $j$ such that all coordinating  (resp. anticoordinating) subpopulations with a threshold less (resp. greater) than that of subpopulation $j$ play strategy $\1$.
 
  In \Cref{ex:abstract}, $\x^*_1$ (resp. $x^*_1)$ is a clean-cut equilibrium point for \eqref{eq:example2} (resp. \eqref{eq: example}) and equals $\mb{c}^{11}$ (resp. ${c}^{11}$),
    {and} $\x^*_3$ (resp. $x^*_3$) is a ruffled anticoordinator-driven equilibrium point for \eqref{eq:example2} (resp. \eqref{eq: example}) and equals $\mb{a}^{12}$ (resp. ${a}^{12}$).
 {Finally}, $\x^*_2$ (resp. $x^*_2$) is a ruffled coordinator-driven equilibrium point for \eqref{eq:example2} (resp. \eqref{eq: example})
    and equals $\mb{o}^{12}$ (resp. ${o}^{12}$).
   From now on, for the sake of readability, we drop the word ``ruffled'' and simply say ``(anti)coordinator-driven.'' 
The following lemma proves that the equilibria of the continuous-time population and abstract dynamics are one of these three population states.
\begin{lemma} \label{lem:eqMixedAbstract}
 The following hold under \Cref{ass:ass1}:
 \begin{enumerate}
  \item The equilibrium points of the dynamics \eqref{eq: type_mixed} and \eqref{eq: abstract_het} are either clean-cut, anticoordinator-driven or coordinator-driven population states.
   \item ${c}^{ij}$ (resp. $\mb{c}^{ij}$) is an equilibrium of \eqref{eq: abstract_het} (resp. \eqref{eq: type_mixed}) iff
    \begin{align} \label{eq:type0}
    \max \{\tau_{i+1}, \tau'_{j} \}< {\pi}'_{j}+ {\pi}_{i} < \min \{\tau_{i}, \tau'_{j+1} \}.
    \end{align}
    \item ${a}^{ij}$ (resp. $\mb{a}^{ij}$) is an          equilibrium \eqref{eq: abstract_het} (resp. \eqref{eq: type_mixed}) iff
    \begin{align} \label{eq:type1}
    0 < \tau_{i} - ({\pi}'_{j} + {\pi}_{i-1}) < {\rho}_{i}.
\end{align}
  \item ${o}^{ij}$ (resp. $\mb{o}^{ij}$) is an          equilibrium \eqref{eq: abstract_het} (resp. \eqref{eq: type_mixed}) iff
    \begin{align} \label{eq:type2}
      0 < \tau'_{j} - ({\pi}_{i} + {\pi}'_{j-1})<  {\rho}'_{j}.
    \end{align}
     \end{enumerate}
\end{lemma}

    According to \Cref{lem:eqMixedAbstract}, for every equilibrium of the continuous-time population dynamics, e.g., $\mb{c}^{ij}$, there is exactly one \emph{associated} abstract equilibrium point, i.e., $c^{ij}$.

\subsection{Global stability analysis} \label{subsec: globalstabilityanalysis}
The abstract dynamics and, in turn, the continuous-time population dynamics may admit several clean-cut, anticoordinator-driven, or coordinator-driven equilibrium points. 

Denote the set of clean-cut, anticoordinator-driven, and coordinator-driven equilibrium points of the continuous-time population dynamics by $\bm{\mathcal{Q}}^c$, $\bm{\mathcal{Q}}^a$, and $\bm{\mathcal{Q}}^o$, respectively.
When there is no equilibrium point of either type, the corresponding set will be empty.
Assume that there are altogether $\mathtt{Q} \in \mathbb{N}$ equilibrium points of clean-cut or anticoordinator-driven type.
Arrange the clean-cut or anticoordinator-driven equilibrium points of the continuous-time population dynamics \eqref{eq: type_mixed}  in the ascending order of their associated abstract equilibrium points.
Hence, we have
$$
  q^*_1 <  q^*_2  < \ldots < q^*_{\mathtt{Q}}, 
$$
where $q^*_k$ is the abstract state at $\mb{q}^*_k \in \bm{\mathcal{Q}}^a \cup \bm{\mathcal{Q}}^c$ for $k \in [\mathtt{Q}]$.
A set $\bm{\mathcal{M}}$ is \emph{attractive} under differential inclusion \eqref{eq_differentialInclusion} from a set $\bm{\mathcal{U}}$ if for each solution $\x(t)$ with $\x(0) \in \bm{\mathcal{U}}$ and each open $\epsilon$-neighborhood  of $\bm{\mathcal{M}}$, there exists time $T>0$ such that $\x(t)$ falls in $\epsilon$-neighborhood of $\bm{\mathcal{M}}$ for all $t \geq T$ \cite{MAYHEW20111045}.
The union of all sets  from which $\bm{\mathcal{M}}$ is attractive under differential inclusion $\bm{\V}$ is the \emph{basin of attraction} of $\bm{\mathcal{M}}$ under $\bm{\V}$ \cite{MAYHEW20111045}.

According to \Cref{lem:betweentwoq}, between every two consecutive anticoordinator-driven and/or clean-cut equilibria $q^*_{k}$ and $q^*_{k+1}$, $k\in [\mathtt{Q}-1]$, there exists exactly one coordinator-driven equilibrium, which we denote by $q^*_{k,k+1}$. 
The associated equilibrium in the continuous-time dynamics is denoted by $\mb{q}^*_{k,k+1}$.
Define $q^*_{0,1} = 0$ and $q^*_{\mathtt{Q},\mathtt{Q}+1} = 1$.
The following theorem reveals the global behavior of the continuous-time population dynamics.
\begin{theorem} \label{thm:ROA}
    Consider the continuous-time population dynamics \eqref{eq: type_mixed} and its associated abstract dynamics \eqref{eq: abstract_het}.
     Under \Cref{ass:ass1}, 
    \begin{enumerate}
        \item  
        Each $\mb{q}^*_{k-1,k}, k \in \{2,\ldots,\mathtt{Q}\},$ is unstable, and 
        each $\mb{q}^*_k, k \in [\mathtt{Q}]$, is 
        asymptotically stable with the basin of attraction
\begin{equation*}
    \bm{\mathcal{A}}(\mb{q}^*_k) 
    =
       \{ \x\in\bm{\X}_{ss} |
       x\in (q^*_{k-1,k}, q^*_{k,k+1})
    \}. 
\end{equation*}
\item 
The limit set of every point in the set $\{\x\in \bm{\X}_{ss} \vert x = {q}^*_{k,k+1} \}$ for $ k \in [\mathtt{Q}-1]$ is   
${\{}\bm{q}^*_{k}, \bm{q}^*_{k+1},\bm{q}^*_{k,k+1}{\}}$.
    \end{enumerate}
\end{theorem}
Indeed, it can be shown that the clean-cut equilibrium points are exponentially asymptotically stable and the anticoordinator-driven equilibrium points are finite-time stable.
\begin{cor}
\label{lem: stabilityOfType01}
   Under \Cref{ass:ass1}, for the continuous-time population dynamics \eqref{eq: type_mixed} and the associated abstract dynamics \eqref{eq: abstract_het},
   \begin{enumerate}
   \item  there exists at least one clean-cut or anticoordinator-driven equilibrium point;
    \item
    if the dynamics admit only one equilibrium point, it is  globally asymptotically stable.
  \end{enumerate}
\end{cor}
\subsection{The Birkhoff center of the continuous-time population dynamics}
The following proposition determines the Birkhoff center of the dynamical system induced by \eqref{eq: type_mixed}.
\begin{prop}\label{prop1}
 If Assumption \ref{ass:ass1} holds, the Birkhoff center of the dynamical system induced by the continuous-time population dynamics \eqref{eq: type_mixed} is 
 \begin{equation}
     \mathrm{BC}_{\text{M}} =  \bm{\mathcal{Q}}^c \cup \bm{\mathcal{Q}}^a \cup \bm{\mathcal{Q}}^o. \label{phiM}
 \end{equation} 
\end{prop}
The following corollaries are direct results of the Proposition 
\ref{prop1} and
determine the Birkhoff centers of the dynamical systems induced by the continuous-time population dynamics capturing the behavior of the populations of all coordinators and all anticoordinators, respectively.
\begin{cor} \label{corCor}
   If Assumption \ref{ass:ass1} holds, the Birkhoff center of the continuous-time population dynamics associated with a population of all coordinators is
   \begin{align}
    \mathrm{BC}_{\text{C}} &=  \bm{\mathcal{Q}}^c \cup  \bm{\mathcal{Q}}^o.  \label{phiC}
 \end{align}
\end{cor}
As for the population of anticoordinators, the result would be $\mathrm{BC}_{\text{A}} =  \bm{\mathcal{Q}}^a \cup  \bm{\mathcal{Q}}^c$.
However, in this case, these sets can be simplified as it can be shown that there is only one equilibrium (\Cref{lem: stabilityOfType01}) and it is either a clean-cut or ruffled anticoordinator-driven.
We summarize this possibility into the following state:
$$\mb{q} = ({\rho}_1, \ldots, {\rho}_{p-1}, \min \{\tau_p -  {\pi}_{p-1}, {\rho}_p \}, 0, \ldots, 0).$$
In view of \Cref{lem:eqMixedAbstract}, the condition for the existence of this equilibrium would be 
$
       {\pi}_{p-1} <  \tau_p$ and $\tau_{p+1} <  {\pi}_{p}
$.
\begin{cor} \label{corAnti}
   If Assumption \ref{ass:ass1} holds, the Birkhoff center of the continuous-time population dynamics describing the evolution of a population of $\p$ types of anticoordinators is
   \begin{align}
     \mathrm{BC}_{\text{A}} &= \{\mb{q} \}.\label{phiA}
 \end{align}
\end{cor}

So far, the steady-state behavior of the continuous-time population dynamics \eqref{eq: type_mixed} was determined.
The following section determines the asymptotic behavior of the discrete population dynamics \eqref{populationDynamicsDiscrete} as the population size approaches infinity. 

\section{The asymptotic behavior of the discrete population dynamics}

 The following theorem asserts the main result of this paper.
\begin{theorem} \label{thm:2}
    Consider the discrete population dynamics {for a population size $\N$} \eqref{populationDynamicsDiscrete}.
    {Let ${\mu}^{\frac{1}{\N}}$} be an
    invariant probability measure of the corresponding population dynamics Markov chain.
    Under \Cref{ass:ass1}, { for every vanishing sequence $\langle \frac{1}{\N} \rangle_{\N \in \mathcal{N}}$ we have $\lim_{ \frac{1}{\N} \to 0} {\mu}^{\frac{1}{\N}}(\bm{\mathcal{O}})=1$ where $\bm{\mathcal{O}}$ is any open set containing }
    \begin{itemize}
    \item $\mathrm{BC}_{\text{M}}$ if the population is a mixture of coordinating and anticoordinating subpopulations;
    \item $\mathrm{BC}_{\text{A}}$ if the population consists of  anticoordinating subpopulations;
    \item  $\mathrm{BC}_{\text{C}}$ if the population consists of  coordinating subpopulations.
    \end{itemize}
\end{theorem}
 In words, consider a population of coordinators and anticoordinators who update their decisions asynchronously based on \eqref{eq: scor} and \eqref{eq: santi}, respectively.
 As the population size approaches infinity, the population state almost surely only visits the states close to the equilibria. 
 Consequently, the amplitude of the possible {perpetual} fluctuations {in the population proportion of $\1$-players} will converge to zero.
 This results in the following corollary.

\begin{cor} \label{cor_fluctuationsDoNotScaleWithN_discretePopulationDynamics_2}
Consider the population dynamics \eqref{populationDynamicsDiscrete}.
    Under \Cref{ass:ass1} and starting from a given initial condition,
   as the population size approaches infinity, the amplitude of the fluctuations in the population proportion of $\1$-players converges to zero with probability one.
\end{cor}
 {
Nevertheless, we can still have perpetual fluctuations (with nonzero probability), but the amplitude of the fluctuations 
will vanish with the population size almost surely.
Roughly speaking, the state if does not equilibrate for ``large'' populations most of the time in the long run will spend ``close'' to the isolated members of the Birkhoff center.}
\begin{remark}
In some setups, such as in \cite{ramazi2020convergence}, agents include themselves when determining their preferred strategy.
Our main results hold for such update rules.
\end{remark}
\setcounter{revisit}{1}
\begin{revisit}
    Consider the associated continuous-time population dynamics with the finite population dynamics introduced in Example \ref{exampleFinite}. 
    Based on \Cref{lem:eqMixedAbstract}, the mean dynamics admit only one equilibrium point at $x = \tau_1$, which is an anticoordinator-driven equilibrium point.
    Therefore, according to \Cref{lem: stabilityOfType01} this equilibrium point is globally asymptotically stable.
    \Cref{fig:finitepopEx2} depicts the evolution of population proportion of strategy-$\1$ players merged with \Cref{fig:finitepop1}.
    As seen, $x(t)$ approaches  $\tau_1$.
    The chattering observed in the solid black line is due to the numerical errors.
\end{revisit}
\begin{example} \label{exampleSpecific}
    Consider a population of size $\N$ consisting of four coordinating and  three anticoordinating subpopulations.
    The population distribution over the subpopulations is $\bm{\rho} = (\tr{2/28}, \tr{3/28}, \tr{3/28}, {6/28}, {8/28}, {3/28}, {3/28})$.
    The thresholds of the anticoordinating subpopulations are $\tr{\tau_1} = \tr{0.929}$, $\tr{\tau_2} = \tr{6/7}$, $\tr{\tau_3} = \tr{0.357}$, and those of coordinators are ${\tau'_1} = {0.05}$, ${\tau'_2} = {0.321}$, ${\tau'_3} = {0.5}$, ${\tau'_4} = {0.643}$.
    When the population size $\N$ equals $28$ or $56$ (which is the next larger valid size), the population proportion of $\1$-players either reaches $x^{\N} = 11/28$ and remains there or fluctuates around $x^{\N} = 6/7$. However, for larger populations, the fluctuation disappears.

    \textbf{Continuous-time population dynamics.} In view of \Cref{lem:eqMixedAbstract}, the continuous-time population dynamics associated with the described discrete population dynamics admit three equilibrium points;
    one clean-cut equilibrium point 
    $\mb{c}^{22} $ at $ (\tr{2/28},\tr{3/28},\tr{0},{0},{0},{3/28},{3/28})$ with the corresponding abstract state $c^{22}$ equals $11/28$, one anticoordinator-driven equilibrium point $\mb{a}^{24}$ at $ (\tr{2/28},\tr{2/28},\tr{0},{6/28},{8/28},{3/28},{3/28})$ with the corresponding abstract state $a^{24}$ equal to $6/7$, and one coordinator-driven equilibrium point $\mb{o}^{23}$ at $(\tr{2/28}, \tr{3/28}, \tr{0}, {0}, {3/28}, {3/28}, {3/28})$ with the corresponding abstract state $o^{23}$ equal to the threshold of the third coordinating subpopulation, $ 0.5$.
    According to \Cref{thm:ROA},
    except for a set of initial conditions of measure zero, the population state converges to either of $\mb{c}^{22}$ or $\mb{a}^{24}$, depending on the initial condition--\Cref{fig:finitepop}.
    As for the discrete population dynamics, the observed fluctuations are around ${a}^{24}$.
    In addition, the equilibrium state at $x^{\N} = 11/28$ preserved as the
    population size grows and is indeed equal to ${c}^{22}$. 
\end{example}
\begin{remark}
Several studies have analyzed the population dynamics by studying one-dimensional dynamics.
In \cite{granovetter1978threshold}, the discrete-time evolution of a heterogeneous population of coordinators was studied by analyzing the evolution of the population proportion of strategy-$\mathtt{A}$ players.
A similar approach was used in \cite[Section 9]{roohi} for a heterogeneous population of coordinators and anticoordinators under a synchronous updating rule, where all agents are active at every time step.
In \cite{comoApproximation,ravazzi2023asynchronous}, scalar continuous-time dynamics instead serve as approximations of the evolution of the population proportion of strategy-$\mathtt{A}$ players in structured populations.
In this paper, however, we were interested in studying the evolution of the proportion of strategy-$\mathtt{A}$ players within each subpopulation.
Hence, a one-dimensional dynamic would not reveal the behavior of each sub-population.
The dimension should be at least as many as the number of sub-populations, i.e., $\p+\p'$.
Moreover, the evolution of the total proportion of $\1$-players at index $k$ depends on the proportions of $\1$-players in subpopulations.
Nevertheless, the expected change of $x^\N$ can be calculated, and that is how we provided the abstract dynamics, which facilitated the subsequent analysis of the $(\p + \p')$-dimensional population dynamics. 
\end{remark}

\begin{figure}
    \centering
    \includegraphics[width=0.8\linewidth]{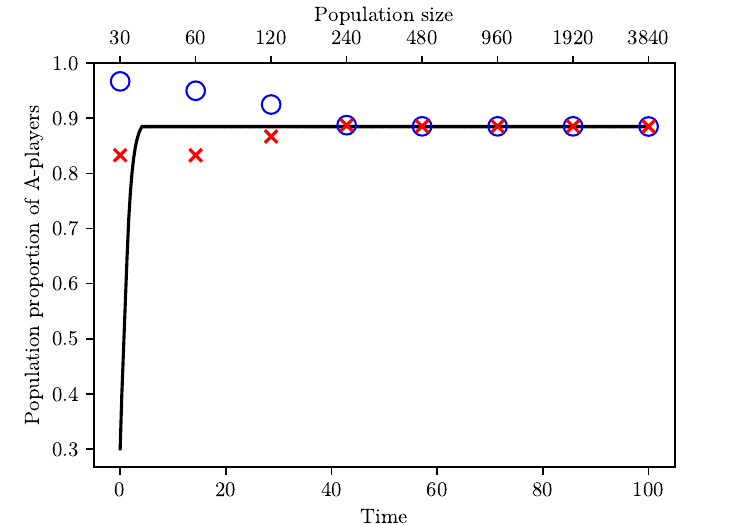}
    \caption{\textbf{The associated abstract state with the discrete population dynamics described in Example \ref{exampleFinite} approaches  $\tau_1$.}
    The black solid line represents the evolution of the abstract state over time. 
    The circles and crosses represent the upper and lower bounds of the invariant sets for different population sizes.
    }
    \label{fig:finitepopEx2}
\end{figure}

\begin{figure}
    \centering
    \includegraphics[width=0.9\linewidth]{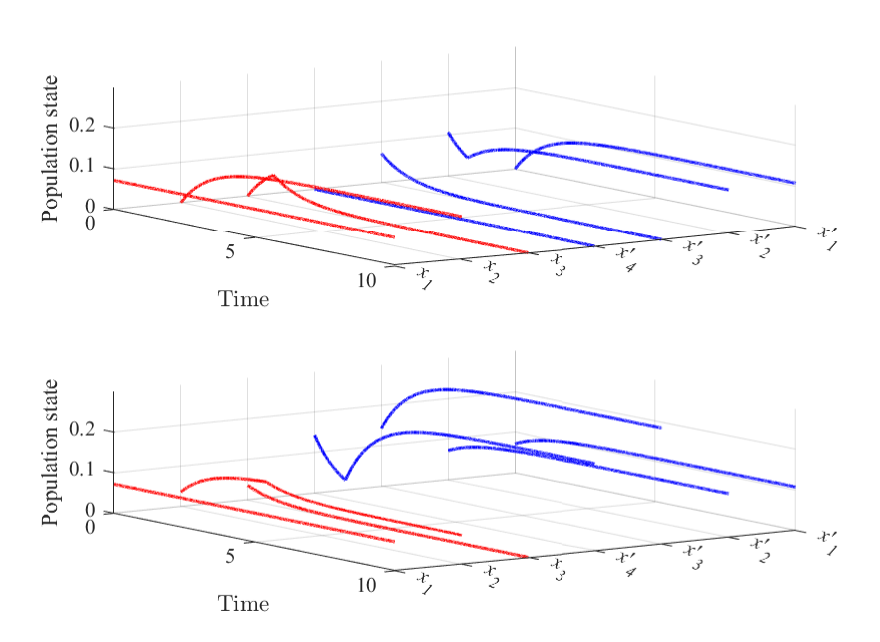}
    \caption{\textbf{Example \ref{exampleSpecific}. The solution of the  continuous-time population dynamics for two different initial conditions.} 
    The upper panel shows the evolution of the population state which converges to clean-cut equilibrium point $\mb{c}^{22} = (\tr{2,3,0},{0,0,3,3})/28$ and the lower one depicts the evolution of the population state converging to the anticoordinator-driven equilibrium point $\mb{a}^{24} = (\tr{2,2,0},{6,8,3,3})/28$, where the corresponding population proportion of $\1$-players equals $\tau_2 = 6/7$.
    The observed non-smoothness in the evolution of some subpopulations is due to the change in their {preferred} strategy.
    For example,
    in the upper panel, at time zero, the abstract state is below  $\tau'_2$  (resp. $\tau_3$), resulting in $\dot{x}'_2 <0$ ($\dot{x}_3 >0$). 
    As abstract state increases and exceeds $\tau'_2$ (resp. $\tau_3$), the {preferred} strategy of coordinating (resp. anticoordinating) subpopulation $2$ (resp. $3$) changes to $\1$ (resp. $\2$) yielding  $\dot{x}'_2 >0$ (resp. $\dot{x}_3 <0$).}
    \label{fig:finitepop}
\end{figure}


\section{Concluding Remarks}
We analyzed the asymptotic behavior of a well-mixed finite heterogeneous population of agents repeatedly choosing between two strategies when the population size approaches infinity.
We utilized the available results linking the stationary measures of the Markov chain corresponding to this discrete population dynamics with the steady-state behavior of the associated continuous-time mean dynamics. 
We then showed that the fluctuations in the finite mixed population of coordinators and anticoordinators {likely} do not {persist} with the population size (\Cref{fig:revisit}).

The convergence result of the continuous-time population dynamics for populations of all anticoordinators, i.e., \Cref{corAnti}, matches that of the discrete dynamics \cite{ramazi2017asynchronous} where the population state was proven to always reach the clean-cut equilibrium $\mb{q}$ or reach the set of two ruffled states (that are not possible here due to the continuity of the state space of the continuous-time population dynamics).
Similarly, the convergence result of the continuous-time population dynamics for  populations of all coordinators, i.e., \Cref{corCor}, matches that of the discrete dynamics \cite{ramazi2020convergence} where the population state was proven to always equilibrate. 
Moreover, the equilibria of the discrete and continuous-time dynamics match.

In addition, the stability of clean-cut equilibrium points of continuous-time population dynamics matches that of the discrete dynamics \cite{le2022heterogeneous}.
However, the stability of the ruffled types has not been investigated in the discrete population.

In the case of the fluctuations, in all examples, we observed that the interval defined by the minimum and maximum of the ratio of strategy $\1$-players in the same population and for a fixed initial condition and different activation sequence contains an anticoordinator-driven equilibrium point of the abstract dynamics.
Moreover, as the population size grows, the interval shrinks. 
It remains unknown whether this observation holds always for every population and whether ruffled anticoordinator-driven equilibria are necessary for the existence of minimal positively invariant sets in the finite discrete populations.
The difficulty partly arises from the inherent challenge of characterizing the positively invariant sets and amplitude of the fluctuations.
 Although \cite[Corollary 6]{roohi} provides bounds on the fluctuations, it still requires the benchmark subpopulations to be known, which appears to be an unsolved complicated problem.
The result of this paper holds in a well-mixed population. 
Whether a similar result can be obtained for structured populations, such as k-regular networks, is an open question.
In this regard,
approaches similar to those in \cite{ravazzi2023asynchronous} might work.

\section*{Acknowledgment}
The authors thank the anonymous reviewers for constructive comments and Hichem Ben-El-Mechaiekh for technical discussions.

\setcounter{figure}{1}
\begin{figure}
    \centering
    \tikzset {_vma69n84m/.code = {\pgfsetadditionalshadetransform{ \pgftransformshift{\pgfpoint{0 bp } { 0 bp }  }  \pgftransformrotate{0 }  \pgftransformscale{2 }  }}}
    \pgfdeclarehorizontalshading{_tlw99ukwl}{150bp}{rgb(0bp)=(0.88,1,1);
    rgb(37.5bp)=(0.88,1,1);
    rgb(39.25bp)=(0.88,1,1);
    rgb(40.5bp)=(0.99,1,1);
    rgb(45bp)=(0.9,0.97,0.99);
    rgb(51bp)=(0.78,0.93,0.98);
    rgb(56.25bp)=(0.75,0.89,0.97);
    rgb(62.5bp)=(0.69,0.85,0.96);
    rgb(100bp)=(0.69,0.85,0.96)}
    \tikzset {_yweun0qfs/.code = {\pgfsetadditionalshadetransform{ \pgftransformshift{\pgfpoint{0 bp } { 0 bp }  }  \pgftransformrotate{0 }  \pgftransformscale{2 }  }}}
    \pgfdeclarehorizontalshading{_xx3bytkm9}{150bp}{rgb(0bp)=(0.88,1,1);
    rgb(37.5bp)=(0.88,1,1);
    rgb(39.25bp)=(0.88,1,1);
    rgb(40.5bp)=(0.99,1,1);
    rgb(45bp)=(0.9,0.97,0.99);
    rgb(51bp)=(0.78,0.93,0.98);
    rgb(56.25bp)=(0.75,0.89,0.97);
    rgb(62.5bp)=(0.69,0.85,0.96);
    rgb(100bp)=(0.69,0.85,0.96)}
    \tikzset {_jaxq4wg33/.code = {\pgfsetadditionalshadetransform{ \pgftransformshift{\pgfpoint{0 bp } { 0 bp }  }  \pgftransformrotate{0 }  \pgftransformscale{2 }  }}}
    \pgfdeclarehorizontalshading{_00fuj0tf3}{150bp}{rgb(0bp)=(0.88,1,1);
    rgb(37.5bp)=(0.88,1,1);
    rgb(39.25bp)=(0.88,1,1);
    rgb(40.5bp)=(0.99,1,1);
    rgb(45bp)=(0.9,0.97,0.99);
    rgb(51bp)=(0.78,0.93,0.98);
    rgb(56.25bp)=(0.75,0.89,0.97);
    rgb(62.5bp)=(0.69,0.85,0.96);
    rgb(100bp)=(0.69,0.85,0.96)}
    \tikzset{every picture/.style={line width=0.75pt}} 
    \begin{tikzpicture}[x=0.75pt,y=0.75pt,yscale=-1,xscale=1]
    
    \draw  [color={rgb, 255:red, 74; green, 144; blue, 226 }  ,draw opacity=1 ][fill={rgb, 255:red, 74; green, 144; blue, 226 }  ,fill opacity=1 ] (134,238) -- (169,228) -- (204,238) -- (186.5,238) -- (186.5,258) -- (204,258) -- (169,268) -- (134,258) -- (151.5,258) -- (151.5,238) -- cycle ;
    \path  [shading=_tlw99ukwl,_vma69n84m] (166.5,123.3) .. controls (227.25,123.3) and (276.5,146.36) .. (276.5,174.8) .. controls (276.5,203.25) and (227.25,226.3) .. (166.5,226.3) .. controls (105.75,226.3) and (56.5,203.25) .. (56.5,174.8) .. controls (56.5,146.36) and (105.75,123.3) .. (166.5,123.3) -- cycle ; 
     \draw  [color={rgb, 255:red, 74; green, 144; blue, 226 }  ,draw opacity=1 ] (166.5,123.3) .. controls (227.25,123.3) and (276.5,146.36) .. (276.5,174.8) .. controls (276.5,203.25) and (227.25,226.3) .. (166.5,226.3) .. controls (105.75,226.3) and (56.5,203.25) .. (56.5,174.8) .. controls (56.5,146.36) and (105.75,123.3) .. (166.5,123.3) -- cycle ; 
    
    \draw  [color={rgb, 255:red, 74; green, 144; blue, 226 }  ,draw opacity=1 ][fill={rgb, 255:red, 74; green, 144; blue, 226 }  ,fill opacity=1 ] (131.5,115.3) -- (149,115.3) -- (149,103.3) -- (184,103.3) -- (184,115.3) -- (201.5,115.3) -- (166.5,123.3) -- cycle ;
    \path  [] (166.5,268.3) .. controls (227.25,268.3) and (276.5,291.36) .. (276.5,319.8) .. controls (276.5,348.25) and (227.25,371.3) .. (166.5,371.3) .. controls (105.75,371.3) and (56.5,348.25) .. (56.5,319.8) .. controls (56.5,291.36) and (105.75,268.3) .. (166.5,268.3) -- cycle ; 
     \draw  [color={rgb, 255:red, 74; green, 144; blue, 226 }  ,draw opacity=1 ] (166.5,268.3) .. controls (227.25,268.3) and (276.5,291.36) .. (276.5,319.8) .. controls (276.5,348.25) and (227.25,371.3) .. (166.5,371.3) .. controls (105.75,371.3) and (56.5,348.25) .. (56.5,319.8) .. controls (56.5,291.36) and (105.75,268.3) .. (166.5,268.3) -- cycle ; 
    
    \path  [shading=_00fuj0tf3,_jaxq4wg33] (166.5,0.3) .. controls (227.25,0.3) and (276.5,23.36) .. (276.5,51.8) .. controls (276.5,80.25) and (227.25,103.3) .. (166.5,103.3) .. controls (105.75,103.3) and (56.5,80.25) .. (56.5,51.8) .. controls (56.5,23.36) and (105.75,0.3) .. (166.5,0.3) -- cycle ; 
     \draw  [color={rgb, 255:red, 74; green, 144; blue, 226 }  ,draw opacity=1 ] (166.5,0.3) .. controls (227.25,0.3) and (276.5,23.36) .. (276.5,51.8) .. controls (276.5,80.25) and (227.25,103.3) .. (166.5,103.3) .. controls (105.75,103.3) and (56.5,80.25) .. (56.5,51.8) .. controls (56.5,23.36) and (105.75,0.3) .. (166.5,0.3) -- cycle ; 
    
    \draw  [color={rgb, 255:red, 74; green, 144; blue, 226 }  ,draw opacity=1 ][fill={rgb, 255:red, 74; green, 144; blue, 226 }  ,fill opacity=1 ] (310.8,264.86) .. controls (310.8,303.24) and (283.94,334.36) .. (250.8,334.36) -- (250.8,316.36) .. controls (283.94,316.36) and (310.8,285.24) .. (310.8,246.86) ;\draw  [color={rgb, 255:red, 74; green, 144; blue, 226 }  ,draw opacity=1 ][fill={rgb, 255:red, 74; green, 144; blue, 226 }  ,fill opacity=1 ] (310.8,246.86) .. controls (310.8,220.51) and (298.15,197.59) .. (279.5,185.8) -- (279.5,179.8) -- (250.8,186.36) -- (279.5,209.8) -- (279.5,203.8) .. controls (298.15,215.59) and (310.8,238.51) .. (310.8,264.86)(310.8,246.86) -- (310.8,264.86) ;
        \draw (298,303.65) node [anchor=north west][inner sep=0.75pt]  [font=\footnotesize,rotate=-0.56] [align=center] {\begin{minipage}[lt]{54.62pt}\setlength\topsep{0pt}
    \begin{center}
    {\fontfamily{ptm}\selectfont continuous-time }\\{\fontfamily{ptm}\selectfont population }\\{\fontfamily{ptm}\selectfont dynamics}
    \end{center}
    
    \end{minipage}};
    \draw (52.26,222.7) node [anchor=north west][inner sep=0.75pt]  [font=\footnotesize,rotate=-0.56] [align=center] {\begin{minipage}[lt]{54.62pt}\setlength\topsep{0pt}
    \begin{center}
    {\fontfamily{ptm}\selectfont Stochastic }\\{\fontfamily{ptm}\selectfont approximation }\\{\fontfamily{ptm}\selectfont theory}
    \end{center}
    
    \end{minipage}};
    \draw (90,305) node [anchor=north west][inner sep=0.75pt]  [font=\footnotesize] [align=center] {{\fontfamily{ptm}\selectfont Coordinating  }\\{\fontfamily{ptm}\selectfont Anticoordinating }\\{\fontfamily{ptm}\selectfont Mixture of both} };
    \draw (182.07,282.03) node [anchor=north west][inner sep=0.75pt]  [rotate=-0.53]  {$t\ \rightarrow \ \infty $};
    \draw (175.08,314.89) node [anchor=north west][inner sep=0.75pt]  [font=\small,rotate=-0.19] [align=left] {{\fontfamily{ptm}\selectfont \textbf{Equilibrates}}\\};
     \draw (170.08,157.3) node [anchor=north west][inner sep=0.75pt]  [font=\footnotesize,rotate=-0.0] [align=center] {{\fontfamily{ptm}\selectfont Equilibrates}\\{\fontfamily{ptm}\selectfont Equilibrates}\\{\fontfamily{ptm}\selectfont The {fluctuations}}
     \\{\fontfamily{ptm}\selectfont {vanish}}};
    \draw (98.91,282.03) node [anchor=north west][inner sep=0.75pt]  [font=\footnotesize,rotate=-0.76] [align=center] {\begin{minipage}[lt]{40.13pt}\setlength\topsep{0pt}
    \begin{center}
    {\fontfamily{ptm}\selectfont \textbf{Population}}
    \end{center}
    \end{minipage}};
    \draw (181,36.3) node [anchor=north west][inner sep=0.75pt]  [font=\footnotesize] [align=center] {{\fontfamily{ptm}\selectfont Equilibrates}\\{\fontfamily{ptm}\selectfont Equilibrates}\\{\fontfamily{ptm}\selectfont May Fluctuate}};
    \draw (290,21.77) node [anchor=north west][inner sep=0.75pt]  [font=\footnotesize,rotate=0] [align=center] {\begin{minipage}[lt]{60pt}\setlength\topsep{0pt}
    \begin{center}
    \centering
    {\fontfamily{ptm}\selectfont Discrete}
    {\fontfamily{ptm}\selectfont population}
    {\fontfamily{ptm}\selectfont of finite size $\N$}
    \end{center}
    \end{minipage}};
    \draw (91,36.3) node [anchor=north west][inner sep=0.75pt]  [font=\footnotesize] [align=center] {{\fontfamily{ptm}\selectfont Coordinating }\\{\fontfamily{ptm}\selectfont Anticoordinating}\\{\fontfamily{ptm}\selectfont Mixture of both} };
     \draw (290,155) node [anchor=north west][inner sep=0.75pt]  [font=\footnotesize,rotate=0] [align=center] {\begin{minipage}[lt]{60pt}\setlength\topsep{0pt}
    \begin{center}
    \centering
    {\fontfamily{ptm}\selectfont Discrete}
    {\fontfamily{ptm}\selectfont population}
    {\fontfamily{ptm}\selectfont $\N \rightarrow \infty$}
    \end{center}
    \end{minipage}};
     \draw (90,157.3) node [anchor=north west][inner sep=0.75pt]  [font=\footnotesize] [align=center]{{\fontfamily{ptm}\selectfont Coordinating }\\{\fontfamily{ptm}\selectfont Anticoordinating}\\{\fontfamily{ptm}\selectfont Mixture of both} };
    \draw (182.07,13.42) node [anchor=north west][inner sep=0.75pt]  [rotate=-0.53]  {$t\ \rightarrow \ \infty $};
    \draw (182.07,135.33) node [anchor=north west][inner sep=0.75pt]  [rotate=-0.53]  {$t\ \rightarrow \ \infty $};
    \draw (95.91,15.89) node [anchor=north west][inner sep=0.75pt]  [font=\footnotesize,rotate=-0.76] [align=center] {\begin{minipage}[lt]{40.13pt}\setlength\topsep{0pt}
    \begin{center}
    {\fontfamily{ptm}\selectfont  \textbf{Population}}
    \end{center}
    
    \end{minipage}};
    \draw (95.91,135.33) node [anchor=north west][inner sep=0.75pt]  [font=\footnotesize,rotate=-0.76] [align=center] {\begin{minipage}[lt]{40.13pt}\setlength\topsep{0pt}
    \begin{center}
    {\fontfamily{ptm}\selectfont  \textbf{Population}}
    \end{center}
    \end{minipage}};
    \end{tikzpicture}

    \caption{\textbf{[Revisit]-The connection between the asymptotic behavior of finite populations of interacting agents evolving based on \eqref{populationDynamicsDiscrete} and their associated continuous-time population dynamics \eqref{eq: type_mixed} based on the stochastic approximation theory.} {As population size approaches infinity, the amplitude of the fluctuations in the population proportion of strategy-$\1$ players vanishes with probability $1$.}
    }
    \label{fig:revisit}
 \end{figure}

\bibliographystyle{IEEEtran}
\bibliography{IEEEabrv,ref}
\section*{{Appendix-Basic Definitions}}
The following are mainly adopted from \cite{klappenecker2018markov}. 
A \emph{discrete-time stochastic process} $\langle\mb{X}_k\rangle_{k=0}^{\infty}$ is a sequence of random variables indexed by $k \in \mathbb{Z}_{\geq 0}$, where $\mb{X}_k$ is the \emph{state of the stochastic process}. 
A \emph{Markov chain} is a discrete-time stochastic process $\langle\mb{X}_k\rangle_{k=0}^{\infty}$ where $\mb{X}_{k+1}$ is independent of the states $\mb{X}_{k-1},\ldots,\mb{X}_0$ given $\mb{X}_k$.
A Markov chain is \emph{finite} if the space over which the states are defined{, denoted by $\bm{\mathcal{S}}$,} is finite.
The transition probabilities of a Markov chain is denoted by $ \text{Pr}_{\x,\y}$ which represents the probability of reaching $\y$ at index $k+1$ given that the state at index $k$ is equal to $\x$, i.e., $\text{Pr}_{\x,\y}=\mathbb{P}[\mb{X}_{k+1} = \y \vert \mb{X}_k = \x]$.
The transition matrix $\mathbf{P} = [\text{Pr}_{\x,\y}]$ is defined accordingly.
A Markov chain is \emph{homogeneous} if its transition probabilities are time independent.
The \emph{invariant probability measure} for a Markov chain with  transition matrix $\mathbf{P}$
{
is a probability measure satisfying $\mu(\x) = \sum_{\y \in \bm{\mathcal{S}}} \text{Pr}_{\x,\y} \mu(\y)$.
}
{
A \emph{realization} of a Markov chain lies in a space constructed by a countable product $\Pi_{i=0}^{\infty}\bm{\mathcal{S}}_i$, where each $\bm{\mathcal{S}}_i$ is a copy of $\bm{\mathcal{S}}$ equipped with a copy of the $\sigma$-algebra of $\bm{\mathcal{S}}$ \cite{Meyn:1253345}.
Intuitively, a realization of a Markov chain is a set of values of  $\mb{X}_0,\mb{X}_1,\ldots$, which are generated based on the transition matrix $\mathbf{P}$.
We say that an event $A$ happens almost surely if it happens with probability one, i.e., $\mathbb{P}(A)=1$. 
This well-known definition is motivated by events $A$ whose complements $A^c$  are non-empty, i.e., $A^c\neq \emptyset$, but have a zero measure.
A sequence of probability measures $\langle{\mu_k}\rangle_{k \in \mathbb{Z}_{\geq 0}}$ defined on $\bm{\mathcal{S}}$ converges weakly to a probability measure $\mu$ if for any bounded continuous function $ {f}:\bm{\mathcal{S}} \to \mathbb{R}$, we have $\int f d\mu_k \xrightarrow{k \to \infty} \int f d\mu$ \cite{klenke2013probability}}.
A sequence $\langle \epsilon_k \rangle_{k=0}^{\infty}$ of positive constants that converges to zero is called \emph{vanishing}.

\begin{thm} \label{thm:sandholm}
    \cite[Theorem 3.5 and Corollary 3.9]{roth2013stochastic}
    For a vanishing sequence of $\epsilon>0$, let 
    $ \langle \langle \mb{X}^{\epsilon}_k\rangle{_{k=0}^\infty}\rangle_{\epsilon >0}$
    be GSAPs for a good upper semicontinuous differential inclusion $\dot{\x} \in \bm{\mathcal{V}}(\x)$.
   Assume that for each $\epsilon$, $\langle \mb{X}^{\epsilon}_k\rangle_{k=0}^{\infty} $
   is a Markov chain and let ${\mu}^{\epsilon}$ be an invariant probability measure of $\langle \mb{X}^{\epsilon}_k\rangle{_{k=0}^\infty} $. Let ${\mu}$ be a  {limit point of $\langle{\mu}^{\epsilon}\rangle_{\epsilon >0}$ in the topology of weak convergence.}
    Then the support  of ${\mu}$ is contained in $\mathrm{BC}(\bm{\Phi})$, where $\bm{\Phi}$ is the dynamical system induced by $\dot{\x} \in \bm{\mathcal{V}}(\x)$.
\end{thm}
\section*{Proofs}
\vspace{-5pt}
\subsection{Proof of \Cref{thm:implicationOFSandholm}}
{
	\begin{proof}
		Assume by contradiction that $\lim_{k \to \infty} {{\mu}}^{\epsilon_k}(\bm{\mathcal{ {O}}}) \not \to 1$.
		This implies that for some positive constant $\delta <1 $, there exists a subsequence $\langle \alpha_k\rangle_{k=1}^{\infty} = \alpha_1, \alpha_2, \ldots$ of $\langle \epsilon_k\rangle_{k=0}^{\infty}$ such that 
		${{\mu}}^{\alpha_k}(\bm{\mathcal{{O}}}) <  \delta$ for each $k$.
		As the state space $\bm{\mathcal{X}}_{0}$ is compact, the family of probability measures defined on the space $\bm{\mathcal{X}}_{0}$ is tight \cite{klenke2013probability}.
		Then, in view of Prokhorov’s theorem, every sequence of the probability measures, including $\langle{\mu}^{\alpha}\rangle$  has a weak limit point, which is a probability measure ${\mu'}$.
		This implies that there exists a subsequence  $\langle \beta_k\rangle_{k=1}^{\infty} = \beta_1, \beta_2, \ldots$ of $\langle \alpha_k\rangle_{k=1}^{\infty}$ such that
		$   \lim \inf_{k \to \infty} {{\mu}}^{\beta_k}(\bm{\mathcal{ {O}}}) \geq {{\mu'}}(\bm{\mathcal{{O}}})$.
		On the other hand, in view of \Cref{thm:sandholm}, the support of every limit point is contained in the Birkhoff center $\mathrm{BC}(\bm{\Phi})$, i.e., ${\mu}'(\bm{\mathcal{O}}) =1$.
		This means that there exists some $k^*$ such that for all $k >k^*$ ${{\mu}}^{\beta_k}(\bm{\mathcal{ {O}}}) >  \delta $, a contradiction.
	\end{proof}
}
\vspace{-10pt}
\subsection{Proof of \Cref{propMarkov}}
\begin{proof}
	To prove, it suffices to show that 
	${\mathbb{P}[\mb{X}_{k+1} = \x^{\N}(k+1) \vert \mb{X}_{k} = \x^{\N}(k)] = \text{Pr}_{\x^{\N}(k),\x^{\N}(k+1)}}$.
	That is to show that the probability the population dynamics \eqref{populationDynamicsDiscrete} reach the state $\y^{\N}$ at index $k+1$ from the state $\x^{\N}$ at index $k$ equals the transition probability $ \text{Pr}_{\x^{\N},\y^{\N}} $ in the Markov chain.
	This is straightforward, because at any state $\x^{\N}$ either the active agent switches to $\1$ or to $\2$ or sticks to her current strategy and no other case is possible. 
	This summarizes the possible four cases in \eqref{eq:markov}.
	In the first case, for example, the probability of an agent that belongs to the subpopulation $p$  switches her strategy from $\2$ to $\1$  equals the probability of drawing such an agent 
	${\rho}_p - x^{\N}_p$  if strategy $\1$ is the {preferred} strategy of subpopulation $p$'s agents {with current strategy $\2$}, i.e., $\smash{s^*(p,\x^{\N}, {2}) = 1}$, and equals zero otherwise,  i.e., $\smash{s^*(p,\x^{\N}, {2}) = 2}$. 
\end{proof}
\subsection{Proof of \Cref{lemGSAPS}}
{First, we introduce the following lemma, which facilitates the proof of \Cref{lemGSAPS}.}
{\begin{lem} \label{lem_minimumDistanceBtwThresholds}
		Under \Cref{ass:ass1},
		let $d$ denote the minimum distance between thresholds, i.e., $d = \min_{i,j \in [\p+\p']} |\tau_i - \tau_j|.$
		Then, for any population size $\N$ satisfying $\N d > 1$ and any $\x \in \bm{\mathcal{X}}_{ss}$, 
		the sum $\bm 1^\top \x$ belongs to at most one interval of the form 
		$[\tau_i, \tau_i + \frac{1}{\N})$ with $i \in [\p + \p']$.
	\end{lem}
	\begin{proof}
		We assume, on the contrary, that there is also a threshold $\tau_j$ where
		$\bm 1^\top \x \in [\tau_j, \tau_j + \frac{1}{\N})$ and that $\tau_i \neq \tau_j$. 
		Without loss of generality, let $\tau_j >\tau_i$ which results in 
		$\tau_j < x < \tau_i + 1/\N$ and, in turn, $\tau_j < \tau_i+ 1/\N $. 
		However, according to the assumption $1/\N < d$, and we reach a contradiction.
\end{proof}}
The proof of \Cref{lemGSAPS} is inspired by the steps taken in \cite[Example 4.1]{roth2013stochastic}.

\begin{proof}
	By considering the sequence $\langle \frac{1}{\N}\rangle_{\N\in\mathcal{N}}$ as a vanishing sequence,
	it suffices to show that there exist $\langle\bm{\mathcal{V}}^{\frac{1}{\N}}\rangle_{{\frac{1}{\N}}}$ and $\langle \mb{U}^{\frac{1}{\N}}_k\rangle_{k}$ such that the conditions in \Cref{defGSAP} are met.
	The first condition of \Cref{defGSAP} {can be easily verified.}
	Let $\bm{\nu}^{\frac{1}{\N}}(\bm{x}^{\N})$ denote the expected increment per time unit of the Markov chain $\langle \mb{X}^{\frac{1}{\N}}_k\rangle_{k}$ when the Markov state is at $\bm{x}^{\N}$.
	Since there are $\N$ transitions per unit time, we have
	$\bm{\nu}^{\frac{1}{\N}}(\bm{x}^{\N}) = \smash{\N \mathbb{E}[\mathbf{X}^{\frac{1}{\N}}_{k+1} - \mathbf{X}^{\frac{1}{\N}}_{k} \mid \mathbf{X}^{\frac{1}{\N}}_k = \x^{\N}]}$.
	The $p^{\text{th}}$ element of $\bm{\nu}^{\frac{1}{\N}}$, $p \in [\p + \p']$, which is the expected increment associated with the subpopulation $p$, equals the sum of the multiplication of each possible change with its probability.
	In view of \eqref{eq:markov}, there are two possible changes, i.e., $\frac{1}{\N}$ step increment or $\frac{1}{\N}$ step reduction, which results in ${\nu}^{\frac{1}{\N}}_p({\x}^{\N}) =\N \big( \frac{1}{\N}(\rho_p - x^{\N}_p)(2-s^*(p,\x^{\N},2))  -\frac{1}{\N} (x^{\N}_p (s^*(p,\x^{\N},1)-1))\big)$ and consequently, 
	\begin{align}
	{\nu}^{\frac{1}{\N}}_p({\x}^{\N}) =&\rho_p(2 - s^*(p,\x^{\N},2))-x^{\N}_p \nonumber \\
	&+ x^{\N}_p\left(s^*(p, \x^\N,2) - s^*(p, \x^\N,1)\right). \label{eq_expected}
	\end{align}   
	The second condition of \Cref{defGSAP} is satisfied by 
	$ \mb{U}^{\frac{1}{\N}}_{k+1} =$ $\smash{
		\N(\mb{X}^{\frac{1}{\N}}_{k+1} - \mb{X}^{\frac{1}{\N}}_{k} -
		\mathbb{E}[\mathbf{X}^{\frac{1}{\N}}_{k+1} - \mathbf{X}^{\frac{1}{\N}}_{k} \mid \mathbf{X}^{\frac{1}{\N}}_k = \x^{\N}])}$
	resulting in
	$\mb{X}^{\frac{1}{\N}}_{k+1} - \mb{X}^{\frac{1}{\N}}_{k} - {\frac{1}{\N}}\mb{U}^{\frac{1}{\N}}_{k+1} = {\frac{1}{\N}}\bm{\nu}^{\frac{1}{\N}}(\x)$.
	Now, we investigate Condition 3.
	To show that Condition 3 is satisfied, for a given $\delta$, we need to provide an $\epsilon_0$ satisfying  \cref{eq_thirdCondition}.
	We propose an $\epsilon_0$ which satisfies $\epsilon_0<\min \{d, \delta\}.$ 
	Then, given \Cref{lem_minimumDistanceBtwThresholds}, for any $\x \in \bm{\mathcal{X}}_{ss}$ there exists at most one threshold, say $\tau_l$, $l \in [\p + \p']$, that satisfies $\bm 1^\top \x \in [\tau_l, \tau_l + \epsilon_0)$, with $\frac{1}{\N} \leq \epsilon_0.$
	In what follows, we show that Condition 3 holds for all  $\x \in \bm{\mathcal{X}}_{ss}$ that satisfy $\exists l \in [\p + \p'] (\bm 1^\top \x \in [\tau_l, \tau_l + \epsilon_0))$;
	the other values of $\x$ can be handled similarly.
	Let  $\varepsilon_x = \bm 1^\top \x - \tau_l$.
	Then given $\bm 1^\top \x \in [\tau_l, \tau_l + \epsilon_0))$, we have
	$\varepsilon_x < \epsilon_0$ and, in turn, $\varepsilon_x < \delta.$ 
	Consider the point $\y$ whose all elements, except some element $j$, are equal to the corresponding elements of $\x$, and the $j^\text{th}$ element, $y_j$, equals $x_j - \varepsilon_x$ (the value of $x_j$ should be greater than or equal to $\varepsilon_x$).
	The term $|\x - \y|$ then equals $\varepsilon_x$, and hence 
	the point $\y$, which satisfies $\bm 1^\top \y = \tau_l$,  belongs to the  $\delta$-neighborhood of $\x$, i.e.,
	$|\x - \y|<\delta$.
	For all $p \neq l$, we then have $\V_p({\y}) = \rho_p(2 - s^*(p,{\y},2)) - {y}_p$ and for $p = l$, $\V_p = [-{y}_p, \rho_p - {y}_p]$.
	On the other hand, in view of \eqref{eq_expected} and 
	$\exists l \in [\p + \p'] (\bm 1^\top \x \in [\tau_l, \tau_l + \epsilon_0))$, the term
	${\nu}^{\frac{1}{\N}}_p({\x})$ for $p = l$ equals $\rho_l - 2 x^{\frac{1}{\N}}_l$ (resp. $0$) if $l$ corresponds to a coordinating (resp. an anticoordinating) subpopulation.
	For $p \neq l$, ${\nu}^{\frac{1}{\N}}_p({\x})$ equals 
	$\rho_p(2 - s^*(p,\x,2))-x_p$, and it can be easily shown that it is also equal to the member of the singleton
	$\V_p(\x)$ defined in
	\eqref{eq: type_mixed}.
	By $\bm z = \bm{\nu}^{\frac{1}{\N}}(\x)$ and  choosing $\bm{v}$ to be a
	member of $\bm{\V}({\y})$ with 
	$v_l({\y}) = \rho_l - 2{y}_l$ (resp. $v_l({\y}) = 0$) if $l$ corresponds to a coordinating (resp. an anticoordinating) subpopulation, 
	the term $\vert \bm{z} - \bm{v}\vert$ equals $|\bm{v}({\y}) -\bm{\nu}^{\frac{1}{\N}}(\x)| = |{\y} - \x |$ which is less than $\delta$. 
	This implies that the relation
	$\inf_{\bm{v}\in \bm{\mathcal{V}}(\y)} \!\vert \bm{z} - \bm{v}\vert< \delta$ holds and, in turn,
	Condition 3 is satisfied.
	As for the last condition,  in view of Example 4.1 in \cite{roth2013stochastic}, it can be shown that
	$\langle \mb{U}^{\frac{1}{\N}}_k\rangle_k$ is a martingale difference sequence.
	On the other hand, for all $k$,
	$\vert \mb{U}^{\frac{1}{\N}}_k \vert \!\leq \!\sqrt{\Sigma_{l=1}^{\p + \p'} (1 + \rho_l)^2}$
	and hence $\!\mb{U}^{\frac{1}{\N}}_k\!$ 
	is uniformly bounded.
	Therefore, the conditions asserted in \cite[Proposition 2.3]{roth2013stochastic} are satisfied, and, consequently, the last condition holds.
\end{proof}
\subsection{Proof of \Cref{lemGSAPS} when agents include themselves} \label{previous_case}
As noted in \Cref{setupRemark}, in some set-ups, agents take into account their current strategies when they want to determine their proffered strategies.
In this case, the update rule \eqref{eq: scor} changes to
\begin{align*}
&\mathtt{s}_{i}(k+1)  
= \begin{cases}
\1& \text{if } x^{\N}(k) \geq {\tau(i)},   \\
\2
& \text{otherwise},
\end{cases}
\end{align*}
and the update rule  \eqref{eq: santi} changes to
\begin{align*}
&\mathtt{s}_{i}(k+1)  
= \begin{cases}
\1& \text{if } x^{\N}(k) \leq {\tau(i)},   \\
\2
& \text{otherwise},
\end{cases}
\end{align*}
where $x^\N$ is the population proportion of $\1$-players.
Accordingly, the function $s^*(\cdot, \cdot, \cdot)$, defined in \Cref{def:discreteDynamics} and used in \Cref{eq:markov}, changes to
\begin{alignat*}{2}
&{\hat{s}}(p,\x^{\mathsf{N}}) 
= \\
&\begin{cases}
{1}& \hspace{-5pt}\text{if } (x^{\mathsf{N}} \leq {\tau_p}  \text{ and } p \leq \p)  \text{ or }
(x^{\mathsf{N}} \geq \tau_{\p + \p' +1 -p}'  \text{ and } p > \p),    \\
{2}&\hspace{-5pt}\text{if } (x^{\mathsf{N}} > {\tau_p} \ \text{ and }  p \leq \p)  
\text{ or }
(x^{\mathsf{N}} < \tau_{\p + \p' +1 -p}'   \text{ and }  p > \p).
\end{cases}
\end{alignat*}
If \Cref{lemGSAPS} remains true, then the main results of the papers are also valid.
Condition 1 is already satisfied.
Let $\bm{\nu}^{\frac{1}{\mathsf{N}}}(\bm{x}^{\mathsf{N}})$ denote the expected increment per time unit $t$ of the Markov chain $\langle \mb{X}^{\frac{1}{\mathsf{N}}}_k\rangle_{k}$ when the Markov state is at $\bm{x}^{\mathsf{N}}$, i.e.,
$\bm{\nu}^{\frac{1}{\mathsf{N}}}(\bm{x}^{\mathsf{N}}) = \smash{\N \mathbb{E}[\mathbf{X}^{\frac{1}{\N}}_{k+1} - \mathbf{X}^{\frac{1}{\N}}_{k} \mid \mathbf{X}^{\frac{1}{\N}}_k = \x^{\N}]}$.
The $p$th element of $\bm{\nu}^{\frac{1}{\mathsf{N}}}$, $p \in [\p + \p']$, reads as ${\nu}^{\frac{1}{\mathsf{N}}}_p({\x}^{\mathsf{N}}) =\rho_p(2 - \hat{s}(p,\x^{\N}))-x^{\N}_p$.
Accordingly, Conditions 2 and 4 can be easily verified.
Regarding Condition 3, 
For each given $\delta$, we need to find an $\epsilon_0$ which satisfies \eqref{eq_thirdCondition}.
In what follows, we show that \eqref{eq_thirdCondition} is satisfied for any arbitrarily chosen $\epsilon_0$.
We consider a selection from \eqref{eq: type_mixed} denoted by $\dot{\x} = \bm{\nu}(\bm{x})$ where $\bm{\nu}$ is a $(\p + \p')$-dimensional function and its $p$th element $ {\nu}_{p}$ equals  $\rho_{p} - x_{p}$ if $\1$ is the {preferred} strategy of subpopulation $p$ and equals  $ - x_{p}$ if the preferred strategy of subpopulation $p$ is $\2$.  
This results in $ {\nu}_{p}(\bm{x}) =\rho_{p} (2-\hat{s}(p,\x))  - x_{p}$.
The expected increment per unit of time, $\bm{\nu}^{\frac{1}{\N}}(\x)$, is the same as the function $\bm{\nu}(\x)$.
As a result, equation \eqref{eq_thirdCondition} and, in turn, Condition 3 are satisfied by selecting $\y$ equal to $\x$, regardless of the value of $\epsilon_0$.
This implies that \Cref{lemGSAPS} holds.
\subsection{Proof of \Cref{prop2}}
\begin{proof}
	It holds that $x = \Sigma_{k=1}^{\p + \p'} x_k$ and $0 \leq x \leq 1$.
	Considering
	\Cref{ass:ass1}, the thresholds of the subpopulations are distinct, building $\p + \p'+1$ disjoint open intervals in the unit interval, where their limit points are either $0$, $1$, or the thresholds. 
	Denote by $\mathcal{T}$ the union of the set $\{0,1\}$ with the set of all subpopulations' thresholds.
	When  $x \notin \mathcal{T}$, there exist $i \in \{0, 1, \ldots, \p+1 \}$ and $j \in \{0, 1, \ldots, \p'+1 \}$ where $x \in (\tau'_j, \tau_i)$ such that $\tau_{i+1} <x < \tau'_{j+1}$. 
	In view of \eqref{eq: type_mixed}, in this case, $\bm{\mathcal{V}}(\x)$ is a singleton and for all $0 <l \leq i$ anticoordinating subpopulations, $\dot{x}_l = \rho_l - x_l$  and for the remaining anticoordinating subpopulations, i.e., $ i <l \leq \p$,
	$\dot{x}_l =  -x_l$.
	Similarly, for coordinating subpopulations,  $\dot{x}_l = - x_l$ for all $\p <l \leq \p + \p' -j+1$ and 
	$\dot{x}_l = \rho_l -x_l$ for all $\p + \p' -j+1< l \leq \p + \p'$.
	By summing over the elements of $\dot{\x}$, we have $\dot{x} = \Sigma_{k=1}^{\p + \p'} \dot{x}_k = {\pi}_i + {\pi}'_j - x$.
	Therefore, for $x \in (\tau'_j, \tau_i)$,
	$\X(x) = \{{\pi}_i + {\pi}'_j - x\}$
	which is the same as the last case of $\X$ in \eqref{eq: abstract_het}.
	When  $x \in \mathcal{T}$,  three cases may happen:  $x = \tau'_j$, $x = \tau_i$, or $x \in \{0,1\}$.
	When  $x = \tau_i$, 
	based on \eqref{eq: type_mixed}, except for $\V_i$, the elements of $\bm{\mathcal{V}}(\x)$ are single-valued and contain either $\rho_l -x_l$ or $-x_l$.
	More specifically,  for  anticoordinating subpopulations, we have
	$\dot{x}_l =\rho_l - x_l$  for $l < i$
	and $\dot{x}_l = -x_l$ for  $i<l \leq \p$.
	As for coordinating subpopulations, assume that the set $\{k \in [\p']\vert  \tau'_k < \tau_i\}$ is not empty and let $j =\max\{k \in [\p']\vert  \tau'_k < \tau_i\}$.
	Then, for $l$ satisfying $\p + \p' -j+1 \leq l \leq \p + \p'$, $\dot{x}_l$ will be equal to $\rho_l - x_l$,  and  {for the remaining subpopulations} will be equal to $ - x_l$.
	If  $\{k \in [\p']\vert  \tau'_k < \tau_i\}$ is empty, then $\dot{x}_l$ for $l > \p$ will be equal to $-x_l$.
	Summing over the elements of $\dot{\x}$ except the $i^{\text{th}}$ element, yields ${\pi}_{i-1} + {\pi}'_{j} - (x-x_i) $.
	As for $l = i$, we have $\dot{x}_i\in [-x_i, \rho_i - x_i]$ yielding $\dot{x} \in ({\pi}_{i-1} + {\pi}'_{j} - x,{\pi}_{i} + {\pi}'_{j} - x)$.
	Therefore, $\X(x) = [{\pi}_{i-1} + {\pi}'_{j} - x,{\pi}_{i} + {\pi}'_{j} - x ]$,
	which is the same as the first case of $\X$ in \eqref{eq: abstract_het}.
	If $x = \tau'_j$, a similar reasoning results in $\X(x) = [{\pi}_{i} + {\pi}'_{j-1} - x,{\pi}_{i} + {\pi}'_{j} - x ]$,
	which is the same as the second case of $\X$ in \eqref{eq: abstract_het}.
	Finally, when $x \in \{0,1\}$, a similar reasoning yields $\dot{x} = {\pi}_{\p} - x$ if $x = 0$ and  $\dot{x} = {\pi}'_{\p'} - x$ if $x = 1$.
	These two are equivalent to the first two cases of $\X$ in \eqref{eq: abstract_het} as  $\tau_{\p + 1} = \tau_0' = 0$ and $\tau'_{\p' + 1} = \tau_0 = 1$.
\end{proof}
\subsection{Proof of \Cref{lem:eqMixedAbstract}}
\begin{proof}
	The proofs are first done for the abstract dynamics.
	We first prove parts 2, 3, and 4. 
	\textit{Part 2)}
	(sufficiency) 
	Should \eqref{eq:type0} hold, it would follow that  $c^{ij}\in (\max\{\tau'_j,\tau_{i+1}\}, \min\{\tau'_{j+1},\tau_i\})$ and based on \eqref{eq: abstract_het}, $0\in \X(c^{ij})$, implying that $c^{ij}$ is an equilibrium of \eqref{eq: abstract_het}.
	(Necessity) 
	Should $c^{ij}$ be an equilibrium of the abstract dynamics, then $0\in \X(c^{ij})$. 
	One of the three cases of $\X$ in \eqref{eq: abstract_het} must be active for $x = c^{ij}$.
	The first two cases of $\X$ cannot be active because of \Cref{ass:ass1}.
	So the third one must hold, implying that $c^{ij} \in (\max\{\tau'_j,\tau_{i+1}\}, \min\{\tau'_{j+1},\tau_i\})$, which is equivalent to \eqref{eq:type0}. 
	\textit{Part 3)} and \textit{Part 4)}.
	Following \textit{Part 2}, the sufficiency and necessity can be concluded.
	\textit{Part 1)} 
	In view of \Cref{ass:ass1}, equilibrium points of clean-cut, anticoordinator-driven, and coordinator-driven types do not match for \eqref{eq: abstract_het}. 
	Hence, an abstract state can at most be only one of these three types.
	We prove by contradiction that an equilibrium point of \eqref{eq: abstract_het} should be either of these three types.
	Suppose $x^*$ is an equilibrium point for \eqref{eq: abstract_het} but is neither clean-cut, anticoordinator-driven, nor coordinator-driven.
	At $x = x^*$ one of the cases of $\X$ in \eqref{eq: abstract_het} should be active.
	If the first case is active, we have  $x^* = \tau_i$ and as it is an equilibrium point  $0 \in \X(x^*)$. This results in ${\pi}_{i-1} + {\pi}'_{j} < \tau_i<{\pi}_{i} + {\pi}'_{j}$, which is equivalent to \eqref{eq:type1}, but we assumed that $x^*$ is not an anticoordinator-driven equilibrium point. This is a contradiction and as a result the first case cannot be active.
	Similarly, {for the remaining two cases, we reach a contradiction, and}
	hence $x^*$ should be either  clean-cut, anticoordinator-driven, or coordinator-driven.  
	Now, we prove the lemma for the equilibrium points of the {continuous-time} population dynamics \eqref{eq: type_mixed}.
	{The sufficiency and necessity of \textit{Part 2)} can be easily shown.}
	\textit{Part 3)} (Sufficiency) 
	Assume that \eqref{eq:type1} holds.
	By plugging $\mb{a}^{ij}$ into \eqref{eq: type_mixed}, we see that  $\mb{0} \in \bm{\mathcal{V}}(\mb{a}^{ij})$ and hence $\mb{a}^{ij}$ is an equilibrium point for \eqref{eq: type_mixed}.
	(Necessity) Should $\mb{a}^{ij}$ be an equilibrium point for \eqref{eq: type_mixed},  $\mb{0} \in \bm{\mathcal{V}}(\mb{a}^{ij})$ and consequently ${0} \in \V_i(\mb{a}^{ij})$.
	This results in $0 \in [-{a}^{ij}_i, \rho_i -{a}^{ij}_i]$ and by plugging ${a}^{ij}_i = \tau_i - ({\pi}_{i-1}+{\pi}'_j)$ into $[-{a}^{ij}_i, \rho_i -{a}^{ij}_i]$, we have 
	$0 \in [-\tau_i + {\pi}_{i-1}+{\pi}'_j, \rho_i -\tau_i + {\pi}_{i-1}+{\pi}'_j]$. This results in 
	${\pi}_{i-1}+{\pi}'_j<\tau_i <  + {\pi}_{i}+{\pi}'_j$ which is equivalent to \eqref{eq:type1}.
	\textit{Part 4)} The sufficiency and necessity can be shown similar to \textit{Part 3}.
	\textit{Part 1)} 
	We prove by contradiction that an equilibrium point of  \eqref{eq: type_mixed} should be either of these three types.
	Suppose $\x^*$ is an equilibrium point for \eqref{eq: type_mixed} but is neither clean-cut, anticoordinator-driven, nor coordinator-driven. 
	Now, based on the value of $x^*$, two cases may happen.
	\textit{Case 1)} $x^*$ equals a threshold. 
	Suppose that $x^*$ equals the threshold of one anticoordinating subpopulation, say $i$.
	From $\mb{0} \in \bm{\mathcal{V}}(\x^*)$ we have $0 \in \V_l(\x^*)$ for $l \in [\p + \p']$, where for $l \neq i$, $\V_l(\x^*)$ is a singleton and therefore to have $0 \in \V_l(\x^*)$,  $x^*_l$ must equal $\rho_l$ for $l<i$ or
	$\p + \p' -j+1 \leq l \leq \p + \p'$, and $x^*_l$ must equal $0$, otherwise, where $j = \max \{k \in [\p'] \vert \tau'_k < \tau_i\}$ and $j =0$ if the set is empty.
	Therefore, $x^* -x^*_i = {\pi}_{i-1} + {\pi}'_{j}$.
	On the other hand, $x^* = \tau_i$ resulting in
	$x^*_i = \tau_i - {\pi}_j' - {\pi}_{i-1}$ which is in turn equal to $a^{ij}_i$.
	This results in that $\x^*$ must be equivalent to $\mb{a}^{ij}$ which is a contradiction. 
	Similarly, the case where $x^*$ equals the threshold of one coordinating subpopulation results in a contradiction. 
	Therefore, $x^*$ cannot be equal to a threshold.
	\textit{Case 2)} $x^*$ is not equal to any thresholds.
	In this case,
	$\bm{\mathcal{V}}(\x^*)$ is a singleton and consequently, $\dot{\x} = \mb{0}$ which results in  $\dot{x} = {0}$. 
	We  showed that the equilibrium points of \eqref{eq: abstract_het} should 
	be either  clean-cut, anticoordinator-driven, or coordinator-driven which is a contradiction.
	In view of \textit{Case 1} and \textit{Case 2}, the assumption that an equilibrium point of \eqref{eq: type_mixed} is neither clean-cut, anticoordinator-driven, nor coordinator-driven reaches a contradiction.
\end{proof}
\subsection{\Cref{lem:betweentwoq}}
\begin{lem} \label{lem:betweentwoq}
	Consider the abstract dynamics \eqref{eq: abstract_het}.
	Between every two consecutive anticoordinator-driven and/or clean-cut equilibria, $q^*_{k-1}$ and $q^*_k, k \in \{2,\ldots,\mathtt{Q}\}$, there exists one coordinator-driven  equilibrium $q^*_{k-1,k}$ satisfying  ${q}^*_{k-1}< q^*_{k-1,k}< {q}^*_{k}$.
\end{lem}
\begin{proof}
	We prove by contradiction.
	Consider two consecutive clean-cut or anticoordinator-driven equilibrium points $q^*_{k-1} < q^*_k$.
	Assume on the contrary that there exist no coordinator-driven equilibrium points  between $q^*_{k-1}$ and $q^*_k$.
	In view of the possible cases of $\X$ in \eqref{eq: abstract_het},
	$\Bar{f}(x) = \max \{y \vert y \in \mathcal{X}(x) \}$ will be negative for $x \in (q^*_{k-1}, d^+_{k-1})$, where $d^+_{k-1}$ is lower bounded by $\min \{\bar{\tau}'_{q_{k-1}}, \bar{\tau}_{q_{k-1}} \}$, i.e., $d^+_{k-1} \geq \min \{\bar{\tau}'_{q_{k-1}}, \bar{\tau}_{q_{k-1}} \}$, where $\bar{\tau}'_{q_{k-1}} = \tau'_{j+1} $, $ \bar{\tau}_{q_{k-1}} = \tau_{i}$ when $q^*_{k-1}$ is clean-cut, i.e., $q^*_{k-1} = c^{ij}$  and $\bar{\tau}'_{q_{k-1}} = {\tau}'_{j+1} $, $\bar{\tau}_{q_{k-1}} = \tau_{i-1}$ when $q^*_{k-1}$ is anticoordinator-driven, i.e., $q^*_{k-1} = a^{ij}$.
	Based on the value of $\min \{\bar{\tau}'_{q_{k-1}}, \bar{\tau}_{q_{k-1}} \}$, two cases may happen:
	\emph{Case 1)} $ \min \{\bar{\tau}'_{q_{k-1}}, \bar{\tau}_{q_{k-1}} \} = \bar{\tau}_{q_{k-1}}$ and consequently for a small enough $\varepsilon > 0$,
	$\Bar{f}(\bar{\tau}_{q_{k-1}}-\varepsilon) <0$.
	The sign of $\Bar{f}(\bar{\tau}_{q_{k-1}} +\varepsilon)$ cannot be positive as $\bar{\tau}_{q_{k-1}}$ is the threshold of an anticoordinating subpopulation and considering the cases of $\X$ in \eqref{eq: abstract_het}, the value of $\dot{x}$ at $x = \tau_i + \varepsilon$ is not greater than its value at $x = \tau_i - \varepsilon$ for any $i \in [\p]$.  
	Therefore, the sign of $\Bar{f}$ does not change around $x = \bar{\tau}_{q_{k-1}}$.
	Actually, no matter how many thresholds of anticoordinating subpopulations (non benchmark subpopulations) are in the interval $(q^*_{k-1}, q^*_k)$, the sign of $\Bar{f}$ will not change in the neighborhood of them. 
	So, in this case,  $\Bar{f}(\bar{\tau}_{q_{k-1}} +\varepsilon)$ remains negative.
	\emph{Case 2)}  $ \min \{\bar{\tau}'_{q_{k-1}}, \bar{\tau}_{q_{k-1}} \} = \bar{\tau}'_{q_{k-1}}$ and consequently for a small enough $\varepsilon > 0$,
	$\Bar{f}({\bar{\tau}'_{q_{k-1}}}-\varepsilon) <0$ and considering the second case of $\X$ in \eqref{eq: abstract_het}, the sign of $\Bar{f}(\cdot)$ changes in a neighborhood of a coordinating subpopulation's threshold  only if $\bar{\tau}'_{q_{k-1}}$ satisfies \eqref{eq:type2}, i.e., $\bar{\tau}'_{q_{k-1}}$ is a
	coordinator-driven equilibrium point. 
	However, we assumed that it is not the case.
	Hence, the thresholds of coordinating subpopulations also do not impact the value of $d^+_{k-1}$.
	Putting the results of \textit{Case 1} and \textit{Case 2} together,
	we conclude that $d^+_{k-1}=q^*_{k}$, i.e., the sign of $\Bar{f}(x)$ remains negative for $x \in (q^*_{k-1}, q^*_k)$.
	On the other hand, considering \eqref{eq:type0}, \eqref{eq:type1}, and the cases of $\X$ in \eqref{eq: abstract_het},  it can be concluded that
	$\underline{f}(x) = \min \{y \vert y \in \mathcal{X}(x) \} >0$ for $x \in (d^-_k,q^*_k)$, where  $d^-_k$ is upper bounded by $ \max \{ \bar{\tau}'_{q'_{k}}, \bar{\tau}_{q_{k}} \}$, where  $\bar{\tau}'_{q'_{k}} = {\tau}'_n $, $\bar{\tau}_{q_k} = \tau_{m+1}$ when $q^*_{k}$ is clean-cut, i.e., $q^*_{k} = c^{mn}$ and similarly $\bar{\tau}'_{q'_{k}} = \tau'_n $, $\bar{\tau}_{q_k} = \tau_{m+1}$ when $q^*_{k}$ is anticoordinator-driven, i.e.,  $q^*_{k} = a^{mn}$
	As $x$ decreases, according to the cases of $\X$ in \eqref{eq: abstract_het}, the change of the sign of $\underline{f}(x)$ from positive to negative can only happen if the flow passes through the coordinating subpopulation's threshold which equals a coordinator-driven equilibrium point, i.e., satisfies \eqref{eq:type2}.
	However,  we assumed it would not be the case.
	Hence, we conclude that the sign of $\underline{f}(x)$ is positive for $x \in (q^*_{k-1}, q^*_k)$.
	This is a contradiction as we already reached the conclusion that the sign of $\Bar{f}(x)$ is negative for $x \in (q^*_{k-1}, q^*_k)$.
	Putting  these arguments together, it is concluded that between two consecutive clean-cut or anticoordinator-driven equilibrium points, ${q}^*_{k-1}, {q}^*_k$, there exists 
	{at least one} coordinator-driven equilibrium point.
	{Similarly, it can be shown that the number of coordinator-driven equilibrium points between two consecutive clean-cut/anti-coordinator-driven equilibria does not exceed one. We denote such equilibrium}
	by ${q}^*_{k-1,k}$.
	Following similar arguments, it is straightforward to show that the leftmost and  rightmost equilibrium points are clean-cut or anticoordinator-driven.
	Up to now, it has been shown that between two consecutive clean-cut or anticoordinator-driven equilibrium points of the abstract dynamics, $q^*_{k-1}$ and $q^*_{k}$, there exists a coordinator-driven equilibrium point, $q^*_{k-1,k}$ such that $q^*_{k-1} < q^*_{k-1,k} < q^*_k$.
	On the other hand, based \Cref{lem:eqMixedAbstract}, there is a one-to-one map between equilibrium points of the abstract dynamics and those of the {continuous-time} population dynamics.
	Therefore,  for each two ordered clean-cut or anticoordinator-driven equilibrium points of the {continuous-time} population dynamics $\mb{q}^*_{k-1}$ and $\mb{q}^*_{k}$, there exists a coordinator-driven equilibrium point
	$\mb{q}^*_{k-1,k}$ such that $\smash{q^*_{k-1} < q^*_{k-1,k} < q^*_{k}}$. 
\end{proof}
\subsection{Proof of \Cref{thm:ROA}}
As for \textit{Part 1}, 
the instability of coordinator-driven equilibria is shown by providing an initial condition from which for a small $\epsilon$, there exists no $\delta$ to satisfy the definition of stability.
The stability of clean-cut and anticoordinator-driven equilibria is first shown for the abstract dynamics using Theorem 1 in \cite{cortes2008discontinuous}.
Then, the definition of asymptotic stability is used to show that of clean-cut and anticoordinator-driven equilibria for the {continuous-time} population dynamics.
\textit{Part 2} of the theorem is shown by using the definition of the limit set.
\begin{proof}
	\textit{Part 1)} 
	Regarding the instability analysis of a coordinator-driven equilibrium point, say $\mb{q}^*_{k-1,k} = \mb{o}^{ij}$,  consider the following initial condition $\x(0) = \mb{o}^{ij} + \delta^* \mb{e}_{\p + \p'-j+1}$ for an arbitrarily small $\delta^*  \in (0,\min \{ \tau'_{j+1}, \tau_{i} \}- \tau'_{j})$.
	The corresponding abstract state
	$x(0)$ will then belong to $(\tau'_j , \min \{ \tau'_{j+1}, \tau_{i} \}) $.
	Pick $\epsilon$ to be equal to $0.5 (\min \{ \tau'_{j+1}, \tau_{i}, {\pi}_i+{\pi}_j' \}- \tau'_{j})$.
	In view of \eqref{eq: type_mixed},
	at such $x(0)$, $\bm{\V}(\x)$ is a singleton.
	Moreover, considering the structure of $\mb{o}^{ij}$, 
	at $x(0)$ we have $\dot{x}_l = 0$  for  $l \in \{1,\ldots, \p + \p'-j, \p + \p'-j+2, \ldots,\p + \p' \}$.
	Indeed, the equation $\dot{x}_l = 0$ holds true as long as $x \in (\tau'_{j},\min \{ \tau'_{j+1}, \tau_{i} \})$.
	On the other hand,
	$x(0) > \tau'_j$ and in view of \eqref{eq:type2}, we have  $\dot{x}_{\p + \p'-j+1}(t) = \rho'_j - x_{\p + \p'-j+1}(t) >0$, for all $t \leq t'$, where $t' =  \inf \{ t\geq 0 \vert x_{\p + \p'-j+1}(t) =  \min \{\alpha_{i,j}, {\rho}'_{j} \} \}$, where
	$\alpha_{i,j} =\min \{\tau'_{j+1}, \tau_{i}\}-({\pi}_i+{\pi}_{j-1})$.
	If $\min \{\alpha_{i,j}, {\rho}'_{j} \} = \alpha_{i,j}$, then $x(t') =\min \{\tau'_{j+1}, \tau_{i}\}$.
	Otherwise, $x(t') = {\pi}_j+{\pi}_i$.
	In either case, we observe that $\vert x(t') -\tau_j' \vert$ is greater than the selected value for $\epsilon$, i.e., $0.5 (\min \{ \tau'_{j+1}, \tau_{i}, {\pi}_i+{\pi}_j' \}- \tau'_{j})$, and
	the value of $\delta^*$ does not impact the value of $x(t')$.
	This indicates instability of ${o}^{ij}$ and  in turn $\mb{o}^{ij}$.
	{The} stability propert{ies} of the other two types of equilibrium points, {are} first  show{n} for the abstract dynamics. 
	Let $V_k(x):\mbb{R}_+\rightarrow \mbb{R}_+ = \frac{1}{2}(x-q^*_k)^2$, where $q^*_k = {\pi}_{i} + {\pi}'_{j}$ if $q^*_k$ is clean-cut, i.e., $q^*_k = c^{ij}$ and  $q^*_k = \tau_{i}$ if it is anticoordinator-driven, i.e., $q^*_k = a^{ij}$.
	The set-valued derivative of $V_k(x)$ w.r.t \eqref{eq: abstract_het} will be $\mathcal{D}(x) = \{(x-q^*_k)\nu | \nu \in \mathcal{X}(x) \}$ \cite{cortes2008discontinuous}.
	We also have $\Bar{f}(x) = \max \{y \vert y \in \X(x) \}< 0$ for $x \in (q^*_k,q^*_{k,k+1})$ and 
	$\underline{f}(x) = \min \{y \vert y \in \X(x) \} > 0$ for $x \in (q^*_{k-1,k},q^*_k)$ (refer to  the proof of \Cref{lem:betweentwoq}).
	Hence, $\max \mathcal{D}(x) <0$ for $x \in (q^*_{k-1,k},q^*_{k,k+1})/ \{q^*_k\}$.
	According to  \cite[Theorem 1]{cortes2008discontinuous}, $q^*_k$ is a strongly asymptotically stable equilibrium point for \eqref{eq: abstract_het}.
	Also the  basin of attraction  of $q^*_k$, will be $(q^*_{k-1,k},q^*_{k,k+1})$.
	As for $q^*_1$ (resp. $q^*_{\mathtt{Q}}$), we have $\Bar{f}(0)>0$ (resp. $\Bar{f}(1)<0$) and this case can be handled similarly which results in
	$\mathcal{A}(q^*_1) = [0, q^*_{1,2})$ (resp. $\mathcal{A}(q^*_{\mathtt{Q}}) = (q^*_{\mathtt{Q}-1,\mathtt{Q}},1]$).
	Now, we show the asymptotic stability of $\mb{q}^*_k$ when it is an anticoordinator-driven equilibrium point, i.e., $\mb{q}^*_k = \mb{a}^{ij}$.
	The case where $\mb{q}^*_k$ is clean-cut can be handled similarly.
	We show that for every $\varepsilon>0$, there exists a $\delta >0$ such that if $\vert \x(0) - \mb{a}^{ij}\vert < \delta$, then $\vert \x(t) - \mb{a}^{ij}\vert < \varepsilon$ for all $t\geq0$, and $\vert \x(t) - \mb{a}^{ij}\vert \rightarrow 0$ as $t \rightarrow \infty$.
	Consider a small enough $\epsilon>0$, such  that at an $\epsilon$-neighborhood of $a^{ij}$, for $l \in [\p + \p']-\{i\}$, $\V_l(\x)$ is a singleton and $\dot{x}_l=\rho_l - x_l$ for $l < i$ or for $l$ satisfying both $\p < l$ and $0 < \p + \p' + 1 -l \leq j$. 
	For the remaining $l$'s, $\dot{x}_l=- x_l$.
	From the stability of the abstract dynamics, the existence of $\delta_x>0$ follows such that  $\vert x(0) - a^{ij}\vert < \delta_x$ results in $\vert x(t) - a^{ij}\vert < \epsilon$ for $t \geq 0$.
	As a result, for $l \in [\p + \p']-\{i\}$, the final value of  $x_l$ is $a^{ij}_l$ and consequently $x_l$ will approach  $a^{ij}_l$ exponentially, i.e.,
	{
		\begin{equation} \label{eq:eq1-proof-Th2}
		\vert x_l(t) - {a}^{ij}_l\vert \leq  (x_l(0) - a^{ij}_l)e^{-t},
		\end{equation}
	}
	for $t \geq 0$.
	On the other hand, we know that for any $y, z \in \mathbb{R}$, $\vert y + z\vert \leq \vert y \vert + \vert z \vert $. 
	Hence, in view of 
	$
	\vert x_i - a^{ij}_i \vert =$  $\vert ( x - a^{ij}) -  \Sigma_{l \neq i} (x_l - a^{ij}_l) \vert
	$
	, we have 
	{
		\begin{equation} \label{eq:eq1-proofTh2}
		\vert x_i - a^{ij}_i \vert \leq \vert ( x - a^{ij}) \vert  + \vert \Sigma_{l \neq i} (x_l - a^{ij}_l) \vert.
		\end{equation}
	}
	Consequently, for $t \geq 0$
	{$$\vert \x(t) - \mb{a}^{ij}\vert^2 <  (\vert x(t) - a^{ij}\vert  +\vert \Sigma_{l \neq i} (x_l - a^{ij}_l) \vert)^2 + \Sigma_{l \neq i} \vert  x_l - a^{ij}_l \vert^2.$$}
	On the other hand, starting from $\delta_x$-neighborhood of $a^{ij}$, {in view of \eqref{eq:eq1-proof-Th2}, we have} 
	{$
		\vert \x(t) - \mb{a}^{ij}\vert^2 <  \big((\p + \p') \max_l \vert  x_l(0) - a^{ij}_l \vert  
		+ (\p + \p' -1)\max_l \vert  x_l(0) - a^{ij}_l \vert e^{-t}\big)^2  +  (\p + \p' -1)\max_l \vert  x_l(0) - a^{ij}_l \vert^2 e^{-2t}
		$
	}
	and{, in turn,}
	$\vert \x(t) - \mb{a}^{ij}\vert$ $< \max_l \vert  x_l(0) - a^{ij}_l \vert \big(\p + \p'  + (\p + \p' -1)e^{-t})^2 + (\p + \p' -1) e^{-2t}$$ \big)^{0.5}. $
	{ If $\delta$ is less than $\delta_x/(\p + \p')$, then in view of the relation between Euclidean norm and Norm-1, $\sum_{l=1}^{\p + \p'} \vert x_l(0) - a_l^{ij} \vert$ will be less than $\delta_x/\sqrt{\p + \p'}$.
		and, in turn, $\vert x(0) - a^{ij}\vert < \delta_x$.}
	Hence, $\mb{a}^{ij}$ is 
	stable because by choosing 
	{$ \delta =\min \{ \frac{\delta_x}{\p + \p'}, \frac{\varepsilon}{\sqrt{(\p + \p')(\p + \p'  + (\p + \p' -1))^2 +( \p + \p')(\p + \p' -1)}\}}, $ }
	we have $\vert x(t) - a^{ij} \vert < \epsilon$  for all $t \geq 0$, and consequently 
	$\vert \x(t) - \mb{a}^{ij}\vert < \varepsilon$.
	As for asymptotic stability, we have
	$\vert \x(t) - \mb{a}^{ij}\vert^2 <  (\vert  x(t) - a^{ij} \vert  + (\p + \p' -1)\max_l \vert  x_l(0) - a^{ij}_l \vert e^{-t})^2 + (\p + \p' -1)\max_l \vert  x_l(0) - a^{ij}_l \vert^2 e^{-2t}$, for $t \geq 0$.
	Since each term in the right hand side approaches zero, we have
	$\lim \vert \x(t) - \mb{a}^{ij}\vert \rightarrow 0$ as $t \rightarrow \infty$.
	Now, we obtain the basin of attraction of $\mb{q}^*_k$.
	\emph{(i)} It is straightforward to show that if $x(0)$ falls in the basin of attraction $q^*_k$, the abstract dynamics will converge to an arbitrary small neighborhood of $q^*_k$ in a finite time.
	Now, depending on the types of $q^*_k$, two cases may happen: \textit{Case 1)} $q^*_k = c^{ij}$, i.e., $q^*_k$ is clean-cut, then let $t = t_0$ be the time {instant} at which the abstract state $x$ enters a small enough $\epsilon$-neighborhood of $q^*_k$, such that  for $l \leq i$ or  $l \geq \max \{ \p + \p'-j+1, \p+1\}$, we have $\dot{x}_l = \rho_l -x_l$ and for the remaining $l$'s, we have $\dot{x}_l =  -x_l$.
	Then for $t > t_0$, the asymptotic value of $x_l$ for $l \in [\p + \p']$ is the same as {the $l$-th entry}  of  $\mb{q}^*_k$.
	Moreover, each $x_l$ approaches $c^{ij}_l$ exponentially.
	As a result, for $t > t_0$ we have
	$\vert \x(t) - \mb{q}^*_k\vert < C \exp(-{(}t{-t_0})) $, for some constant $C$.
	\textit{Case 2)} $q^*_k = a^{ij}$, i.e., $q^*_k$ is anticoordinator-driven.
	Then, a time moment similar to $t_0$ defined in \textit{Case 1} exists such that
	for $l \in [\p + \p']-\{i\}$,  $\vert x_l(t) - a^{ij}_l\vert < C \exp(-{(}t{-t_0})) $, for some constant $C$ and $t \geq t_0$.
	Therefore, {in view of \eqref{eq:eq1-proofTh2}}, it is straightforward to show that for any arbitrary small $\epsilon_1$, {there exists}  some finite time $t_1 + t_0$, {such that} 
	$\vert x_i(t) -  {q}^*_{k_i}\vert < \epsilon + (\p + \p' -1)\epsilon_1$ for $t > t_1+t_0$ resulting in $\vert \x(t) - \mb{q}^*_k \vert < (\epsilon + (\p + \p' -1)\epsilon_1)^2 + (\p + \p' -1 )\epsilon_1^2$. 
	Having this and \emph{(i)} we conclude that 
	in both cases 1 and 2, for any  arbitrary $\varepsilon$-neighborhood of $\mb{q}^*_k$, time $T>0$ can be found such that for all $t \geq T$, $\vert \x(t) - \mb{q}^*_k\vert < \varepsilon$.
	As a result, the basin of attraction of point $\bm{q}^*_{k}$ equals $\{ \x \in \bm{\X}_{ss} \vert x \in (q^*_{k-1,k}, q^*_{k,k+1})\}$.
	The basins of attraction of $\mb{q}^*_1$ and $\mb{q}^*_{\mathtt{Q}}$ can be obtained similarly.
	\textit{Part 2)} When $x(0) = q^*_{k-1,k}$, {where} $q^*_{k-1,k} = o^{ij}$,
	two cases may happen.
	\textit{Case 1)} The abstract state $x$ remains at $q^*_{k-1,k}$.
	Then for $l \neq j$, $\V_l(\x)$ is a singleton and $\dot{x}_l$ is  $\rho_l - x_l$  for $l \leq i$, or for $l > \p$ and $\p + \p'+1-j < l$, and $\dot{x}_l$  is  $- x_l$ for the remaining $l$'s.
	As a result, for $l \in [\p + \p']-\{j\}$, the final value of  $x_l$ is $o^{ij}_l$ and consequently $x_l$ will approach  $o^{ij}_l$ exponentially. 
	Then for any arbitrary small $\varepsilon$, in a finite time $T$ we have $|x_l(t) - o^{ij}_l| < \varepsilon$  for $t \geq T$ and subsequently $\vert x_j(t) - o^{ij}_j \vert < \vert x(t) - o^{ij} \vert + \Sigma_{l \neq j}  \vert (x_l(t) - o^{ij}_l)\vert \leq (\p + \p' -1)\varepsilon$ resulting in
	$\vert \x(t) - \mb{o}^{ij} \vert \leq \varepsilon \sqrt{(\p + \p')(\p + \p'-1)}$ for $t \geq T$. 
	Hence,
	$\mb{o}^{ij}$ is the limit point. 
	\textit{Case 2)} The abstract state $x$ leaves  $q^*_{k-1,k}$. Then, the provided reasoning for the case $x \neq q^*_{k-1,k}$  will be applicable resulting in $ \mb{q}^*_{k-1}$ or $\mb{q}^*_{k}$ to be the limit point.
	Overall, the limit set of $\mb{q}^*_{k-1,k}$ is $\{ \mb{q}^*_{k-1}, \mb{q}^*_{k-1,k}, \mb{q}^*_{k}\}$. 
\end{proof}
\subsection{Proof of Proposition \ref{prop1}}
\begin{proof}
	Based on \Cref{thm:ROA}, the set $ \mathbf{Q}^c \cup \mathbf{Q}^a \cup \mathbf{Q}^o$ contains the limit sets of all $\x$ in $\bm{\X}_{ss}$.
	To find the corresponding Birkhoff center, we need to find  the closure of the set of recurrent points.
	Based on \Cref{def_Brikhoff}, the set of recurrent points is also equal $ \mathbf{Q}^c \cup \mathbf{Q}^a \cup \mathbf{Q}^o$.
	As a results, the Birkhoff center of the population dynamics will be equal to  $ \mathbf{Q}^c \cup \mathbf{Q}^a \cup \mathbf{Q}^o$.
\end{proof}
\vspace{-20pt}
\subsection{Proof of \Cref{thm:2}}
\begin{proof}
	\emph{(i)} Based on \Cref{propMarkov}, 
	the sequence  $\langle \x^{\N}(k) \rangle$ which evolves according to discrete population dynamics \eqref{populationDynamicsDiscrete},
	is a realization of population dynamics Markov chain $\langle \mb{X}^{\frac{1}{\N}}_k \rangle_k$ with transition probabilities formulated in \eqref{eq:markov}.
	\emph{(ii)} Based on \Cref{lemGSAPS}, the collection of   
	$\langle \langle \mb{X}^{\frac{1}{\N}}_k\rangle_{k} \rangle_{\N \in \mathcal{N}}$
	is a GSAP for {the good upper semicontinuous differential inclusion} \eqref{eq: type_mixed}.
	\emph{({iii})}
	The transition probabilities of the population dynamics Markov chain \eqref{eq:markov} are homogeneous and the state space over which the Markov chain is defined is finite {for each $\N$}.
	This results in the existence of invariant probability measures ${\mu}^{\frac{1}{\N}}$ for Markov chain $\langle\mathbf{X}^{\frac{1}{\N}}_k\rangle_k$ \cite{lalley}.
	\emph{({iv})}  \Cref{prop1} specifies the Birkhoff centers of \eqref{eq: type_mixed}.
	{Consider the sequence $\langle \frac{1}{\N}\rangle_{\N \in \mathcal{N}}$ as a vanishing sequence.}
	\Cref{thm:implicationOFSandholm} and \emph{(i)}-\emph{({iv})} together complete the proof.
\end{proof}
\vspace{-20pt}
\subsection{Proof of \Cref{cor_fluctuationsDoNotScaleWithN_discretePopulationDynamics_2}}
\begin{proof}
	Assume on the contrary that there exists some $N>0$, such that for every population size $\N\in\mathcal{N}_{\geq N}$, where $\mathcal{N}_{\geq N} \subseteq \{\N|\N\in\mathcal{N},\N\geq N\}$ and 
	unbounded,
	the evolution of the proportion of $\1$-players admits some perpetual fluctuation set ${\mathcal{Y}}^\N\subset (0,1) \cap \frac{1}{\N}\mathbb{Z}$ of amplitude $y^\N$, 
	where $y^\N\geq L $ for some positive value $L$, 
	\begin{equation} \label{eq_cor_12}
	\exists L>0 \forall \N\in\mathcal{N}_{\geq N} \qquad y^\N\geq L
	\end{equation}
	and that for every $K\in(0,L]$ there exists some small enough $\delta>0$ such that in the long run the proportion of $\1$-players visits every connected subset of amplitude $K$ of ${\mathcal{Y}}^\N$ (for some $K>0$) with probability greater than $\delta$:
	\begin{align} 
	\forall K>0\exists \delta>0 \forall & \mathcal{Z}^\N
	\subset[\min\mathcal{Y}^\N,\max\mathcal{Y}^\N] \text{ with } |\mathcal{Z}^\N| = K \nonumber\\
	&  \text{ and $\mathcal{Z}^\N$ connected,}  \qquad
	\mu_x^{\frac{1}{\N}}(\mathcal{Z}^\N)>\delta, \label{eq_cor_11}
	\end{align}
	where $\mu_x^{\frac{1}{\N}}$ is 
	defined as $\mu^{\frac{1}{\N}} \circ  f^{-1}$, where $ f: \mathbb{R}^{\p + \p'} \to \mathbb{R},$ and $ f(\x) = \sum_{p=1}^\p x_p.$
	Let $\mathtt{n}$ denote the total number of abstract equilibrium points, and let $q_p$, for $p \in [\mathtt{n}]$, denote the $p^\text{th}$ equilibrium point of the abstract dynamics.
	Given \Cref{thm:2}, it follows that for every vanishing sequence $\langle \frac{1}{n} \rangle = \frac{1}{n_0}, \frac{1}{n_1}, \ldots$ where $n_i \in \mathcal{N}$ for $i \in \mathbb{Z}_{\geq 0}$ 
	$ \lim_{i\to\infty}   \mu_x^{\frac{1}{n_i}} (\bigcup_{p=1}^{\mathtt{n}}{\mathcal{B}}({q_p},\epsilon/2)) = 1,
	$
	for a small enough positive $\epsilon < L/\mathtt{n}$.
	Equivalently, 
	for every $\delta>0$,
	there exists $i^*>0$ such that for all $i > i^*$, 
	$  \smash{\mu_x^{\frac{1}{n_i}} (\bigcup_{p=1}^{\mathtt{n}}{\mathcal{B}}({q_p},\epsilon/2)) > 1- \delta}$
	and, in turn, $\smash{\mu_x^{\frac{1}{n_i}}([\min\mathcal{Y}^\N,\max\mathcal{Y}^\N]-\bigcup_{p=1}^{\mathtt{n}}{\mathcal{B}}({q_p},\epsilon/2)) < \delta}$.
	On the other hand, in view of \eqref{eq_cor_12}, $[\min\mathcal{Y}^\N,\max\mathcal{Y}^\N]-\bigcup_{p=1}^{\mathtt{n}}{\mathcal{B}}({q_p},\epsilon/2)$ contains a connected interval of amplitude at least $ (L-\mathtt{n}\epsilon)/(\mathtt{n}+1)$, 
	This contradicts \eqref{eq_cor_11} for $K=(L-\mathtt{n}\epsilon)/(\mathtt{n}+1)$.
\end{proof}
\subsection{\Cref{lem:discretedynamics}}
\begin{lem} \label{lem:discretedynamics}
	{
		The population dynamics governed by update rules \eqref{eq: scor} and \eqref{eq: santi} and the activation sequence are equivalent to the dynamics defined in  \eqref{populationDynamicsDiscrete}.
	}
\end{lem}
\begin{proof}
	{
		As the update mechanism is asynchronous, the population states at two consecutive time indices $\x^\N(k)$ and $\x^\N(k+1)$ can differ by at most $\pm\frac{1}{\N}$.
		Which element of the population state $\x^\N(k+1)$ may differ from that of $\x^\N(k)$ is determined by the random variable ${P}_k$ which is the active subpopulation at time index $k$, and can take values in $[\p + \p']$.
		The probability of ${P}_k = p$ equals $\rho_p$.
		The amount of change in the ${P}_k$th element of $\x^\N(k)$ depends on the current strategy of the active agent and {her} preferred strategy.
		The former is captured by the random variable ${S}_k$ taking values $1$ and $2$, and the latter is determined by the function $s^*({P}_k, \x^\N(k), {S_k})$ which is obtained from \eqref{eq: scor} and \eqref{eq: santi}.
		The probability of ${S}_k = 1$ (resp. ${S}_k = 2$) given the active subpopulation is $p$ is  $x^\N_p/\rho_p$ (resp. $1-x^\N_p/\rho_p$).
		If the current strategy and the preferred strategy are equal, the population states at time indices $k$ and $k+1$ are the same, and if the current strategy is $\2$ and the preferred one is $\1$, the ${P}_k$th element of $\x^\N(k)$ increases by $\frac{1}{\N}$, and the ${P}_k$th element of $\x^\N(k)$ decreases by $\frac{1}{\N}$ otherwise.
		This update rule is equivalent to \eqref{populationDynamicsDiscrete}.
	}
\end{proof}
\begin{appendixCorollary} \label{lem_other_tie_breaking}
	Assume that an active agent $i$ 
	chooses her strategy uniformly at random when
	$x_{-i}^\N(k) = \tau(i)$. 
	Then, 
	\Cref{lemGSAPS} remains true.
\end{appendixCorollary}
\begin{proof}
	It is easy to show that Conditions 1, 2, and 4 in Definition 4 are satisfied.
	Now, we investigate Condition 3.
	Function $s^*(\cdot, \cdot, \cdot)$ changes to the random variable $S^*$ with the following distribution
	\begin{alignat*}{2}
	\mathbb{P}[&S^*(p,\x^{\N},s) = 1 \mid \x^{\N}] =\\
	&\begin{cases}
	\frac{1}{2} & \text{if } p \leq \p, x^{\N} = \tau_p, \text{ and } s = 2, \\
	& \text{or } p \leq \p, x^{\N} = \tau_p + \frac{1}{\N}, \text{ and } s = 1, \\
	& \text{if } p > \p,  {x}^{\N} = \tau'_{\p+\p'+1-p},\text{ and } s = 2, \\
	& \text{if } p > \p,  {x}^{\N} = \tau'_{\p+\p'+1-p} + \frac{1}{\N} ,\text{ and } s = 1, \\
	1 & \text{if } p \leq \p \text{ and } x^{\N} < \tau_p, \\
	& \text{if } p > \p, x^{\N} > \tau'_{\p+\p'+1-p},\text{ and } s = 2, \\
	& \text{if } p > \p, x^{\N} > \tau'_{\p+\p'+1-p} + \frac{1}{\N},\text{ and } s = 1,\\ 
	0 & \text{if } p \leq \p,  x^{\N} > \tau_p,\text{ and } s = 2, \\
	& \text{if } p > \p \text{ and }  x^{\N} < \tau'_{\p+\p'+1-p}, \\
	& \text{if } p \leq \p,  x^{\N} > \tau_p + \frac{1}{\N},\text{ and } s = 1, 
	\end{cases}
	\end{alignat*}
	\begin{alignat*}{2}
	\mathbb{P}[&S^*(p,\x^{\N},s) = 2 \mid \x^{\N}]\\
	&\begin{cases}
	\frac{1}{2} & \text{if } p \leq \p, x^{\N} = \tau_p,\text{ and } s = 2, \\
	& \text{if } p > \p,  x^{\N} = \tau'_{\p+\p'+1-p},\text{ and } s = 2, \\
	& \text{if } p \leq \p, x^{\N} = \tau_p + \frac{1}{\N},\text{ and } s = 1, \\
	& \text{if } p > \p,  x^{\N} = \tau'_{\p+\p'+1-p} + \frac{1}{\N},\text{ and } s = 1,\\
	0 & \text{if } p \leq \p \text{ and } x^{\N} < \tau_p, \\
	& \text{if } p > \p,  x^{\N} > \tau'_{\p+\p'+1-p} + \frac{1}{\N},\text{ and } s = 1,\\
	& \text{if } p > \p,  x^{\N} > \tau'_{\p+\p'+1-p},\text{ and } s = 2, \\
	1 & \text{if } p \leq \p, x^{\N} > \tau_p,\text{ and } s = 2, \\
	& \text{if } p > \p \text{ and }  x^{\N} < \tau'_{\p+\p'+1-p},\\
	& \text{if } p \leq \p, x^{\N} > \tau_p + \frac{1}{\N},\text{ and } s = 1.\\
	\end{cases}
	\end{alignat*}
	Following the steps taken in the proof of \Cref{lemGSAPS}, the expected change in the population state
	$\bm{\nu}^{\frac{1}{\N}}(\x)$ can be obtained, where for all $\x \in \bm{\mathcal{X}}_{ss}$ satisfying $\nexists \tau_p (\bm{1}^\top\x \in [\tau_p,\tau_p + \frac{1}{\N}])$ is  independent of the tie-breaking rule. 
	Out of those states satisfying $\exists \tau_p (\bm{1}^\top\x \in [\tau_p,\tau_p + \frac{1}{\N}])$, we investigate the state $\x$ with $\bm 1^\top \x = \tau_p$ for some $p$; the other values of $\x$ can be handled similarly.
	The $p^{\text{th}}$ element of $\bm{\nu}^{\frac{1}{\N}}(\x)$, which corresponds to subpopulation $p$, at $\x$ with $\bm 1^\top \x = \tau_p$ equals 
	$
	\N \left( \frac{1}{\N} (\rho_p - x_p)\frac{1}{2} -\frac{1}{\N}  x_p  \times \big(0 \big)\right) = \frac{\rho_p}{2} - \frac{x_p}{2}, 
	$
	if $p$ corresponds to an anticoordinating subpopulation and 
	$
	\N \left( \frac{1}{\N} (\rho_p - x_p)\frac{1}{2} -\frac{1}{\N}  x_p  \big(1\big)\right) = \frac{\rho_p}{2} - \frac{3}{2}x_p, 
	$
	otherwise.
	In both cases the expected change in subpopulation $p$ belongs to $\mathcal{V}_p(\x) = [-x_p, \rho_p - x_p]$ (equation \eqref{eq: type_mixed}).
	For elements $l\neq p$ of $\bm{\nu}^{\frac{1}{\N}}(\x)$, the expected change is equal to the unique member of the $\V_l(\x)$.
	This implies that equation \eqref{eq_thirdCondition} is satisfied with $\y = \x$.
	Hence, Condition 3 holds, and this completes the proof.
\end{proof}
\end{document}